\documentclass[12pt,letterpaper]{article}
\usepackage[dvipsnames]{xcolor}
\usepackage[margin=1in,footskip=0.5in]{geometry}
\usepackage[T1]{fontenc}
\usepackage{verbatim}
\usepackage{amsthm}
\usepackage{amsmath}
\usepackage{amssymb}
\usepackage{caption}
\usepackage{natbib}
\usepackage{soul}
\usepackage{bbm}
\usepackage{setspace}
\usepackage{subfigure}
\usepackage{xcolor}
\usepackage{graphics}
\usepackage{epsfig}
\usepackage{geometry}
\usepackage{accents}

\usepackage{enumitem}
\usepackage{setspace}
\usepackage{array}
\newcolumntype{x}[1]{>{\centering\arraybackslash\hspace{0pt}}p{#1}}
\usepackage{tikz}
\usetikzlibrary{patterns}
\usetikzlibrary{patterns.meta}
\pgfdeclarepattern{
  name=hatch,
  parameters={\hatchsize,\hatchangle,\hatchlinewidth},
  bottom left={\pgfpoint{-.1pt}{-.1pt}},
  top right={\pgfpoint{\hatchsize+.1pt}{\hatchsize+.1pt}},
  tile size={\pgfpoint{\hatchsize}{\hatchsize}},
  tile transformation={\pgftransformrotate{\hatchangle}},
  code={
    \pgfsetlinewidth{\hatchlinewidth}
    \pgfpathmoveto{\pgfpoint{-.1pt}{-.1pt}}
    \pgfpathlineto{\pgfpoint{\hatchsize+.1pt}{\hatchsize+.1pt}}
    \pgfpathmoveto{\pgfpoint{-.1pt}{\hatchsize+.1pt}}
    \pgfpathlineto{\pgfpoint{\hatchsize+.1pt}{-.1pt}}
    \pgfusepath{stroke}
  }
}

\tikzset{
  hatch size/.store in=\hatchsize,
  hatch angle/.store in=\hatchangle,
  hatch line width/.store in=\hatchlinewidth,
  hatch size=5pt,
  hatch angle=0pt,
  hatch line width=.5pt,
}
\usepackage{pifont}

\usepackage{avant}

\newcolumntype{L}[1]{>{\raggedright\let\newline\\\arraybackslash\hspace{0pt}}m{#1}}
\newcolumntype{C}[1]{>{\centering\let\newline\\\arraybackslash\hspace{0pt}}m{#1}}
\newcolumntype{R}[1]{>{\raggedleft\let\newline\\\arraybackslash\hspace{0pt}}m{#1}}

\makeatletter \renewenvironment{proof}[1][\proofname]
{\par\pushQED{\qed}\normalfont\topsep6\p@\@plus6\p@\relax\trivlist\item[\hskip\labelsep\bfseries#1\@addpunct{.}]\ignorespaces}{\popQED\endtrivlist\@endpefalse} \makeatother

\theoremstyle{plain}

\newtheorem{lemma}{Lemma}

\newtheorem{rmk}{Remark}
\newtheorem{corollary}{Corollary}
\newtheorem{proposition}{Proposition}

\usepackage{multirow}
\usepackage{chngcntr}
\usepackage{apptools}
\AtAppendix{\counterwithin{lemma}{section}}
\AtAppendix{\counterwithin{proposition}{section}}

\definecolor{uncblue}{RGB}{75, 165, 211}
\definecolor{nberblue}{RGB}{0, 90, 155}
\definecolor{resgreen}{RGB}{34, 84, 67}
\definecolor{pennblue}{RGB}{0, 44, 119}
\definecolor{pennred}{RGB}{152, 30, 50}

\usepackage{hyperref}
\hypersetup{
hypertexnames=false,
pdftitle={Robust Capital Budgeting},
pdfauthor={Yunzhi Hu, Fei Li, Hanqing Ye, Dihan Zou},
pdfborder = {0 0 0},
urlbordercolor = {0 0 0},
colorlinks=true,
linkcolor=pennblue,
urlcolor=nberblue,
citecolor=pennblue}

\newcounter{parentnumber}

\begin{document}
\newcommand{\regret}{\ensuremath{\mathcal{R}}}

\newcommand{\LB}{\ensuremath{\text{LB}}}

\newcommand{\wcr}{\ensuremath{\overline{\mathcal{R}}}}

\newcommand{\menu}{\ensuremath{\mathcal{K}}}

\newcommand{\envir}{\ensuremath{\mathcal{E}}}

\newcommand{\capex}{\ensuremath{\overline{k}}}

\newcommand{\R}{\ensuremath{\mathbb{R}}}

\providecommand{\capex}{\ensuremath{\overline{k}}}

\title{Robust Capital Budgeting}
\author{Yunzhi Hu \and Fei Li \and Hanqing Ye \and Dihan Zou\thanks{
Hu: UNC Kenan-Flagler Business School (\texttt{Yunzhi\_Hu@kenan-flagler.unc.edu});
Li: Department of Economics, University of North Carolina at Chapel Hill (\texttt{lifei@email.unc.edu});
Ye: Department of Economics, University of North Carolina at Chapel Hill (\texttt{hanqing@unc.edu});
Zou: Department of Economics, University of North Carolina at Chapel Hill (\texttt{dihan@email.unc.edu}).
}}
\date{\today}

\maketitle

\begin{abstract}
\noindent
We study capital budgeting when a manager privately knows project productivity and her empire-building benefits, while headquarters lacks credible priors. Headquarters minimizes maximum regret relative to informed investment. Without compensation, the optimal policy is a spending cap balancing empire-building overinvestment against profitable investment forgone. Compensation helps when empire-building incentives are weak. The optimal payment incentivizes restraint, declines with investment, and reaches zero at the spending cap, supporting greater managerial discretion. When empire-building incentives are strong, compensation cannot distinguish productive investment from biased expansion; therefore, a spending cap alone is optimal.
\end{abstract}

\newpage

\section{Introduction}

Capital budgeting is one of the central problems in corporate finance. It often involves long-lived projects whose value depends on cash flows far into the future. Evaluating these projects therefore requires forecasts precisely where firms have the least confidence in them. Indeed, \citet{Graham2022} reports that, although many companies prepare five- or ten-year plans, CFOs regard only about the first two years of their forecasts as reliable. This observation raises a natural question: how should a firm conduct long-term capital budgeting when uncertainty about the future makes it difficult to formulate credible beliefs?

Passing up an unusually profitable investment can be a consequential mistake. Reflecting on Berkshire Hathaway's capital allocation, Charlie Munger writes that ``almost all huge errors were in not making a purchase.''\footnote{Charles T. Munger, Berkshire Hathaway 2014 Annual Report, p.~42, \href{https://www.berkshirehathaway.com/letters/2014ltr.pdf}{available here}. A recent Invesco commentary makes the same point in different terms: waiting for uncertainty to disappear can be costly because prices and opportunities may adjust before the risks become clear. See Brian Levitt, ``FOMO isn't about ignoring risk,'' Invesco, April 28, 2026, \href{https://www.invesco.com/us/en/insights/fomo-market-risk-investing-missing-out.html}{available here}.} This observation highlights a concern that goes beyond the conventional NPV criterion: the potential magnitude of profits forgone when investment opportunities are rejected. 
An overly restrictive capital budget can therefore 
prevent the firm from exploiting a highly productive opportunity. This concern is especially salient for start-ups, new technologies, drug development, and other exploratory projects whose returns are highly uncertain but potentially exceptional. Public discussions often describe the desire to avoid this ex-post regret as fear of missing out, or FOMO.\footnote{The UK Financial Conduct Authority discusses FOMO in the setting of risky start-ups and new cryptoassets presented as scarce, once-only opportunities. It also warns investors that they may lose all the capital committed and should consider whether the party providing information has motives that are not aligned with their interests. See ``Hype---spot the signs and manage your FOMO,'' updated January 19, 2026, \href{https://www.fca.org.uk/investsmart/hype-spot-signs-manage-your-fomo}{available here}. We use FOMO only as shorthand for concern about forgoing a rare, high-upside opportunity, without requiring social comparison or imitation.} Yet an overly permissive capital budget can lead to overinvestment because of well-documented managerial empire-building incentives, creating the opposite regret when the project performs poorly. Balancing these two concerns depends on how frequently different investment opportunities arise. However, forming such expectations can be difficult, especially in areas unfamiliar to headquarters. How should headquarters design its capital budget without reliable beliefs about the opportunities it may face?

To answer this question, we build a canonical capital budgeting model in which a project manager privately knows the project's expected returns at different investment levels, as well as her own empire-building benefits. Headquarters knows the range of possible returns and private benefits but lacks reliable priors over either. It designs the capital budgeting process to minimize the maximum regret it could incur across all possible investment environments. Here, regret is defined  in the sense of \citet{Savage1951}, and it measures the payoff lost relative to what headquarters could have achieved under complete information. The criterion captures the headquarters' aversion to lost opportunities arising from limited information. We show that, without compensation, the optimal policy is a spending cap that balances the loss from forgone investment profits against that from the manager's empire building. Furthermore, carefully designed compensation can improve on this policy only when incentives for empire building are sufficiently low. In this case, headquarters can maintain a higher spending cap and give the manager more discretion.

Specifically, we introduce a model with one headquarters and one project manager. The manager privately observes how project output and her private benefit vary with the amount of capital invested. By contrast, headquarters knows only that both functions are nondecreasing and concave and that their marginal increments satisfy known upper bounds. Unlike in a standard Bayesian framework, headquarters has no prior over either function. Instead, it chooses the investment levels available to the manager and the compensation attached to those choices to minimize maximum regret.

First, we study the model in which headquarters cannot offer additional compensation, which corresponds to the capital-{rationing} problem in \citet{HarrisRaviv1996}. We show that the optimal policy is a simple spending cap: the manager has full discretion to invest up to a fixed limit. She reaches the cap when either project returns or private benefits make expansion sufficiently attractive, but leaves part of the budget unspent otherwise. The choice of this spending cap captures a familiar tradeoff as in the delegation literature between commitment and flexibility. The optimal limit balances the worst loss from excessive spending against that from profitable investment forgone. This limit goes up when the project can be more productive and falls when potential private benefits increase. The simplicity of this policy is consistent with \citet[p.~1981]{Graham2022}, who suggests that ``[c]onservative and simple rules may work as well as any'' in a complex environment. It also gives formal content to his observation that conservative slack may reflect managers' recognition that tail risks cannot be precisely anticipated in either likelihood or magnitude \citep[p.~2034, n.~57]{Graham2022}.

Next, we consider the role of compensation. We first show that, when incentives for empire-building are strong, compensation cannot further reduce worst-case regret below what headquarters achieves with the optimal spending cap alone. The reason is that an unproductive manager with large private benefits can value investment exactly as a productive manager without such bias does, so compensation cannot distinguish their motives. Moreover, persuading a manager with strong private benefits to invest below the cap requires compensation at least as large as the resulting reduction in the firm's overinvestment loss. By contrast, when empire-building incentives are weak, optimally designed compensation can help. Under the optimal policy, the manager is rewarded for low investment with higher compensation; this compensation decreases with investment and reaches zero when the spending cap is met. The manager therefore gives up compensation when she expands the project. A sufficiently productive project makes this sacrifice worthwhile, whereas a manager with a weak project is induced to preserve capital. Headquarters can consequently allow a higher spending cap and give the manager more discretion while reducing worst-case regret.

\paragraph{Related literature.} 
We contribute to the literature on capital budgeting with asymmetric information. \citet{HarrisRaviv1996} study capital allocation to a manager who privately knows project productivity and enjoys private benefits from investment; their optimal procedure combines a spending limit with costly auditing of requests for additional funding. \citet{HarrisRaviv1998} extend this analysis to allocating resources across projects. \citet{BernardoCaiLuo2001,BernardoCaiLuo2004} jointly study capital allocation and compensation in the presence of private information and managerial effort, emphasizing performance-related pay. Whereas this literature largely adopts a Bayesian mechanism-design approach, our main departure is that headquarters cannot formulate a credible prior and instead minimizes maximum regret.

Our paper belongs to the literature on robust mechanism design that uses minimax regret to limit the worst-case loss relative to what could have been achieved had the underlying environment been known. \citet{BergemannSchlag08} study monopoly pricing without a prior over buyers' valuations, where regret balances missed sales against revenue forgone from low prices. \citet{GuoShmaya2023} study regret-minimizing project choice within organizations, and \citet{GuoShmaya2025} apply the criterion to monopoly regulation without known demand and cost functions. We extend this approach to capital budgeting, where regret reflects investment distortions and compensation paid to an informed manager. To our knowledge, despite a substantial literature in finance on robust design under max--min criteria, including \citet{MalenkoTsoy}, ours is the first to study robust design in finance using minimax regret.

Our spending-cap result also connects the paper to the delegation literature initiated by \citet{Holmstrom1984}. \citet{AlonsoMatouschek2008} characterize optimal delegation to an informed, biased agent and give conditions under which the principal offers an interval of decisions. Our spending cap preserves the manager's ability to use private information while limiting biased expansion, and we establish its optimality among arbitrary closed sets of investment choices. \citet{Frankel2014} and \citet{AlonsoGanHu2026} study delegation robust to uncertainty about the agent's preferences while retaining a prior over external states. We instead allow uncertainty about both productive opportunities and private benefits, jointly designing managerial discretion and nonnegative compensation.

\medskip    

The remainder of the paper is organized as follows. Section \ref{sec:model} presents the model and the contracting reduction. Section \ref{sec:control} studies capital rationing without compensation. Section \ref{sec:joint} characterizes the joint choice of investment authority and compensation and develops its comparative statics. Section \ref{sec:extensions} discusses some extensions. Proofs are in the appendix.

\section{Model}\label{sec:model}\label{maintext:start}

\subsection{Investment, information, and managerial preferences}
Consider a firm with two dates, $t=0$ and $1$. At date 0, headquarters allocates $k\in[0,K]$ to a project, where $K>0$ is the available capital budget. At date 1, the project produces $q(k)$ and generates operating profit
\[
\pi(k;q):=\frac{q(k)}{1+r}-k,
\]
where $r \in (0,1)$ is the cost of capital, which we also refer to as the discount rate throughout the paper.
Given our focus on capital budgeting, the interval between two dates is long, potentially spanning many years. Besides operating profit, the project manager receives a private benefit $b(k)$ from project scale, capturing empire-building incentives.

The manager privately observes $q(\cdot)$ and $b(\cdot)$. Headquarters knows neither function and has no prior over them. It knows only that both functions are nonnegative, nondecreasing, and concave and satisfy
\begin{equation}\label{eq:bounds}
0\leq q(k)\leq Ak,\qquad 0\leq b(k)\leq Bk
\quad\text{for all }k\in[0,K],
\end{equation}
where $A>1+r$ and $B>0$ are common knowledge.  
Concavity captures diminishing returns to productive investment and empire building.
Write the present value of maximum return as $\hat A:=\frac{A}{1+r} > 0$.  Let $\envir$ denote the set of admissible pairs $(q,b)$. These restrictions imply $q(0)=b(0)=0$ and bound marginal increments in output and private benefits by $A$ and $B$, respectively. Both functions are continuous and differentiable almost everywhere.

\subsection{Capital budgeting and direct investment contracts}\label{sec:contracts}
The manager's private information consists of the two functions $q(\cdot)$ and $b(\cdot)$. Designing a capital budgeting procedure therefore appears to be an infinite-dimensional mechanism-design problem. In general, such a procedure allows the manager to submit a message $m$ from a set $M$. A message may be a funding request, a business plan, or a detailed account of the project. In response, headquarters authorizes investment $\kappa(m)$ and pays compensation $\omega(m)\geq 0$. Formally, a mechanism $(M,\kappa,\omega)$ consists of a message set $M$, an investment rule $\kappa:M\to[0,K]$, and a compensation rule $\omega:M\to\R_+$. Headquarters commits to these rules before the manager chooses her message. The payoffs associated with the mechanism are
\begin{align*}
U_P(m;q,b)&=\pi(\kappa(m);q)-\omega(m),\\
U_A(m;q,b)&=\pi(\kappa(m);q)+b(\kappa(m))+\omega(m).
\end{align*}
Given $(q,b)$, the manager chooses a message that maximizes $U_A$. 
We assume the manager's outside option is zero and that some message leads to zero investment. Together with limited liability, this guarantees participation.
We call a mechanism \textit{regular} if its available investment--payment pairs form a nonempty compact set.\footnote{
That is, $\mathcal O:=\{(\kappa(m),\omega(m)):m\in M\}$
is a nonempty compact subset of $[0,K]\times\R_+$.
Compactness bounds payments, includes limiting outcomes,
and guarantees the existence of the manager's best response.
Because some message leads to zero investment and payments
are nonnegative, we may also assume $\mathbf{0}\in\mathcal O$
without loss.
}

Directly optimizing over arbitrary message spaces leads to unnecessary complication. The standard revelation principle offers little simplification because the manager's private type consists of a pair of functions. The next lemma reduces its dimension. The result is related to the taxation principle \citep[pp.~257--258]{FT1991game}, under which a mechanism can be represented by a tariff over allocations. Limited liability requires a slightly different representation here. Because compensation must be nonnegative, headquarters cannot generally exclude an investment through the payment schedule alone while preserving the mechanism's payoffs. A contract must therefore specify both the permitted investment set $\menu$ and the compensation $w(k)$ attached to each permitted investment. The informed manager then chooses among these investment--compensation pairs.

Formally, let
\[
\menu:=\kappa(M) \text{\quad and \quad }
w(k):=\max_{m:\,\kappa(m)=k}\omega(m) \text{ for }k\in\menu.
\]
The set $\menu$ contains all investment levels that can result from some message. For each investment $k\in\menu$, $w(k)$ is the largest compensation offered by any message that leads to $k$.
A \textit{direct investment contract}, or simply a \textit{contract}, is a nonempty closed set $\menu\subseteq[0,K]$ of permitted investments that contains $0$, and a finite, nonnegative, upper-semicontinuous payment function $w:\menu\to\R_+$.\footnote{
Upper semicontinuity means that $w(k)\geq\limsup_n w(k_n)$ whenever $k_n\to k$ within $\menu$. Together with compactness of $\menu$, it guarantees that the manager attains a best choice. Appendix \ref{app:lemma} verifies this property for the construction above and gives a regular mechanism representation for every contract in this class.
} The payoffs associated with a contract become
\begin{align*}
U_P(k;q,w)&:=\pi(k;q)-w(k),\\
U_A(k;q,b,w)&:=\pi(k;q)+b(k)+w(k).
\end{align*}
The informed manager chooses
\begin{equation}\label{BR}\tag{BR}
k\in\arg\max_{x\in\menu}U_A(x;q,b,w).
\end{equation}

\begin{samepage}
\begin{lemma}[Payoff equivalence]\label{lem_direct}
For every regular mechanism, there is a direct investment contract that induces exactly the same set of payoff pairs for headquarters and the manager for every admissible pair $(q,b)$, as the manager ranges over all her best responses. Conversely, every direct investment contract has a regular mechanism representation with this property. 
\end{lemma}
\end{samepage}
The intuition is as follows. If two messages lead to the same investment, every manager prefers the one that pays more. The lower-paying message can therefore be removed without changing any best response or either party's payoff. Retaining only the highest payment at each investment reduces an arbitrary reporting procedure to $(\menu,w)$. Thus, headquarters need not elicit either private function. Instead, it specifies the available investments and the compensation attached to each, and the manager uses her information to select an investment. Compactness and upper semicontinuity imply that the relevant maxima are attained. The proof in Appendix \ref{app:lemma} also establishes the converse.

\begin{rmk}[No output-contingent compensation]
We restrict compensation to depend on the budgeting decision rather than realized output, which we denote by $y=q(k)$. This is because output may be realized many years after the investment decision, whereas compensation may need to be determined and paid much earlier. Output may also be difficult to verify or contract upon, especially for new technologies, exploratory projects, and investments whose benefits arise elsewhere in the firm. This restriction is also important for Lemma \ref{lem_direct}: at a given investment, all managers rank two messages by their compensation alone. 
\end{rmk}

\subsection{Headquarters' objective: minimax regret}\label{sec:regret}
Headquarters chooses a contract without knowing either the production function $q$ or the manager's private-benefit function $b$. Because it has no prior over $(q,b)$, it cannot rank contracts by expected profit. A popular alternative is maximizing the headquarters' worst-case profit. However, given that the admissible set includes an entirely unproductive project, $q(k)=0$ for all $k$, the result would be trivial. No investment will ever be made, regardless of how productive the project may turn out to be.

Motivated by the FOMO discussion, we evaluate a contract relative to the profit headquarters could earn if it knew the production function and chose investment itself. Formally, this informed-investment benchmark is
\[
V(q):=\max_{k\in[0,K]}\pi(k;q).
\]
Note that this benchmark payoff only depends on $q$, and that it is nonnegative because zero investment is feasible.

We measure the performance of a contract by its \textit{regret} in the sense of \citet{Savage1951} and \citet{Milnor1954}: the profit headquarters loses relative to the informed benchmark. For a contract $(\menu,w)$ and a managerial best response $k$, the regret is
\[
\regret(k,q):=V(q)-U_P(k;q,w)=V(q)-\pi(k;q)+w(k).
\]
This loss includes both excessive investment and profitable investment forgone, as well as the compensation paid to the manager. Unlike worst-case profit, regret distinguishes an appropriate decision not to invest in an unproductive project (low regret) from a costly failure to invest in a productive one (high regret). 

Headquarters evaluates a contract by its largest regret across all admissible pairs $(q,b)$ and all managerial best responses:\footnote{
The inner maximum imposes adverse tie breaking: if several investments maximize the manager's payoff, headquarters evaluates the contract using the one that produces the greatest regret.
}
\begin{equation}\label{eq:worst-regret}
\wcr(\menu,w):=
\max_{(q,b)\in\envir}\ 
\max_{k\text{ satisfying }\eqref{BR}}
\{\regret(k, q)\}.
\end{equation}
 Headquarters then solves the minimax-regret problem
\begin{equation}\label{directDesign}\tag{P}
\min_{\menu,w}\wcr(\menu,w),
\end{equation}
over the contracts defined in Section \ref{sec:contracts}. 

Headquarters has two instruments. It can ration the capital available to the manager and condition compensation on the chosen investment. More stringent capital rationing limits empire building but may constrain productive projects. Compensation can reward restraint 
but it is costly and may also discourage productive investment. Each instrument generates regret in two directions.

\section{Least Regret Capital Rationing}\label{sec:control}

\subsection{Optimal Spending Cap}

We start with the case in which compensation is unavailable, so $w\equiv0$. This isolates the role of capital rationing. It is a natural benchmark when project managers receive fixed salaries, when compensation is determined separately from individual investment decisions, or when unused project funds cannot be paid to the manager as a bonus.
 By Lemma \ref{lem_direct}, headquarters chooses a nonempty closed menu $\menu$ of permitted investments, and the manager selects
\[
k_A\in\arg\max_{k\in\menu}\left\{\frac{q(k)}{1+r}-k+b(k)\right\}.
\]
Intuitively, headquarters' choice of capital trades off two sources of regret. A restrictive menu may prevent the manager from exploiting a productive opportunity, whereas a permissive menu may allow private benefits to induce excessive investment. Let us illustrate using two corner cases. Specifically, let
\[
\alpha:=\hat A-1>0,\quad \beta:=\min\{B,1\}.
\]

\paragraph{No capital rationing.} If the manager can choose any $k\in[0,K]$, headquarters' worst-case regret is 
\[
\wcr([0,K],0)=\beta K.
\]
For $B\geq1$, the admissible pair $q(k)=0$ and $b(k)=Bk$ makes the manager exhaust the budget even though informed headquarters would invest nothing. Headquarters loses the investment $K$. For $B<1$, headquarters' loss arises under $q(k)=(1+r)(1-B)k$ and $b(k)=Bk$: the manager is indifferent across investments, headquarters loses $Bk$ at investment $k$, and adverse tie breaking selects $K$. 

\paragraph{No investment.} If headquarters may permit only $k=0$, its worst-case regret is 
\[
\wcr(\{0\},0)=\alpha K.
\]
This is attained by $q(k)=Ak$ and informed headquarters would invest the entire budget. Intuitively, this policy eliminates exposure to excessive investment but forgoes every productive opportunity.

\paragraph{Interval menus.} Headquarters may, in principle, choose any closed menu, including one with gaps. We show that it is without loss of generality to focus on interval menus of the form $[0,\capex]$. To see this, let $\capex$ be the maximum investment in any menu. Under $q(k)=Ak$ and $b(k)=0$, the manager chooses $\capex$, while informed headquarters would invest $K$, generating regret $\alpha(K-\capex)$. Under the unproductive project described above---$q(k)=0$ and $b(k)=Bk$ when $B\geq1$, and $q(k)=(1+r)(1-B)k$ and $b(k)=Bk$ when $B<1$---the manager again chooses $\capex$, while informed headquarters would invest zero, generating regret $\beta\capex$. Hence, every menu with maximum $\capex$ has worst-case regret of at least
\[
\max\{\alpha(K-\capex),\beta\capex\}.
\]

We now show that the interval menu $[0, \capex]$ attains this lower bound. Let $k_P$ attain the informed benchmark $V(q)$, and let $k_A$ be the manager's choice from $[0,\capex]$. Suppose first that $k_P\leq\capex$. Because $k_P$ is available, the manager's optimality implies
\[
\pi(k_A;q)+b(k_A)\geq\pi(k_P;q)+b(k_P) \Rightarrow \pi(k_P;q)-\pi(k_A;q)\leq b(k_A)-b(k_P).
\]
If $k_A<k_P$, monotonicity of private benefits makes the right-hand side nonpositive. Because $k_P$ maximizes headquarters' profit, the left-hand side is nonnegative. Therefore, both sides must equal zero, so headquarters incurs no regret.
 If $k_A\geq k_P$, the above  inequality becomes \[
 \pi(k_P;q)-\pi(k_A;q)\leq B(k_A-k_P),
 \]
 using the bound on private benefits. Meanwhile, by monotonicity of the output, 
\[
\pi(k_P;q)-\pi(k_A;q) \leq k_A - k_P,
\]
i.e., the maximum profit difference is the saved investment cost. 
Thus, regret is at most $\beta\capex$. Suppose instead that $k_P>\capex$. Because $\capex$ is available, managerial optimality and the monotonicity of private benefits imply
\[
\pi(k_A;q)\geq\pi(\capex;q).
\]
The bound on the growth of output then gives
\[
\pi(k_P;q)-\pi(k_A;q)
\leq\pi(k_P;q)-\pi(\capex;q)
\leq\alpha(K-\capex).
\]
Combining the two cases with the lower bound above yields
\[
\wcr([0,\capex],0)=\max\{\alpha(K-\capex),\beta\capex\}.
\]
Therefore, it is without loss of generality to focus on interval menus $[0,\capex]$. It remains only to choose $\capex$. Given that the first loss decreases with the cap, whereas the second increases with it, the maximum is minimized where the two losses are equal. We have the following result.

\begin{proposition}\label{prop_delegation}
When compensation is infeasible, an optimal capital menu is $\menu_0^*=[0,\capex_0]$, where
\begin{equation}\label{eq:cap-zero}
\capex_0=\frac{\alpha}{\alpha+\beta}K
=\frac{\hat A-1}{\hat A+\min\{B-1,0\}}K.
\end{equation}
Its worst-case regret is
\begin{equation}\label{eq:regret-zero}
\regret_0^*=\beta\capex_0
=\alpha(K-\capex_0).
\end{equation}
\end{proposition}

Figure \ref{fig:cap-regret-tradeoff} illustrates the construction. A lower cap sacrifices too much productive investment, whereas a higher cap permits too much loss from empire building. The least-regret cap balances these two losses.

\begin{figure}[htbp]
\centering
\begin{tikzpicture}[x=10cm,y=4cm,font=\footnotesize]
  \def\capFigAlpha{1.2}
  \def\capFigBeta{0.8}
  \pgfmathsetmacro{\capFigOpt}{\capFigAlpha/(\capFigAlpha+\capFigBeta)}
  \pgfmathsetmacro{\capFigRegret}{\capFigBeta*\capFigOpt}
  \tikzset{
    cap productive/.style={draw=pennblue,line width=0.9pt,dash pattern=on 4pt off 3pt},
    cap unprofitable/.style={draw=orange!85!black,line width=0.9pt,dash pattern=on 4pt off 3pt},
    cap envelope/.style={draw=black,line width=1.8pt,line join=round},
    cap guide/.style={draw=black!40,densely dotted,line width=0.6pt}
  }

  \draw[cap productive] (0.03,-0.4)--(0.10,-0.4);
  \node[anchor=west,align=left] at (0.12,-0.4)
    {Productive project\\$\alpha(K-\capex)$};
  \draw[cap unprofitable] (0.60,-0.4)--(0.67,-0.4);
  \node[anchor=west,align=left] at (0.69,-0.4)
    {Unproductive project\\$\beta\capex$};

  \draw[->,line width=0.7pt] (0,0)--(1.08,0) node[right] {$\capex$};
  \draw[->,line width=0.7pt] (0,0)--(0,1.34) node[above] {Regret};
  \node[below left] at (0,0) {$0$};
  \draw (1,0.015)--(1,-0.015) node[below=3pt] {$K$};
  \draw (0.006,\capFigAlpha)--(-0.006,\capFigAlpha)
    node[left=3pt] {$\alpha K$};
  \draw (0.006,\capFigBeta)--(-0.006,\capFigBeta)
    node[left=3pt] {$\beta K$};

    \draw[cap guide] (0, \capFigBeta) -- (1, \capFigBeta);

    \draw[cap guide] (1, 0) -- (1, \capFigBeta);

  \draw[cap productive] (0,\capFigAlpha)--(1,0);
  \draw[cap unprofitable] (0,0)--(1,\capFigBeta);
  \draw[pennred, opacity = 0.5, line width = 3 pt] (0,\capFigAlpha)
    --(\capFigOpt,\capFigRegret)--(1,\capFigBeta);
  \node[anchor=south,align=center] at (0.78,1.05) {Worst-case regret};
  \draw[->,black!65,line width=0.6pt]
    (0.84,1.03)--(0.88,{0.88*\capFigBeta});

  \draw[cap guide] (0,\capFigRegret)--(\capFigOpt,\capFigRegret)
    --(\capFigOpt,0);
  \node[left=4pt] at (0,\capFigRegret) {$\regret_0^*$};
  \fill (\capFigOpt,\capFigRegret) circle[radius=2.5pt];
  \draw (\capFigOpt,0.015)--(\capFigOpt,-0.015);
  \node[below=4pt,align=center] at (\capFigOpt,0)
    {$\capex_0=\dfrac{\alpha}{\alpha+\beta}K$};
\end{tikzpicture}
\caption{Regret under a spending cap.
{\footnotesize A productive project generates underinvestment regret $\alpha(K-\capex)$, while an unproductive project generates overinvestment regret $\beta\capex$. The bold upper envelope is the worst-case regret of $([0,\capex],0)$, minimized at $\capex_0$, where the two losses coincide.}}
\label{fig:cap-regret-tradeoff}
\end{figure}
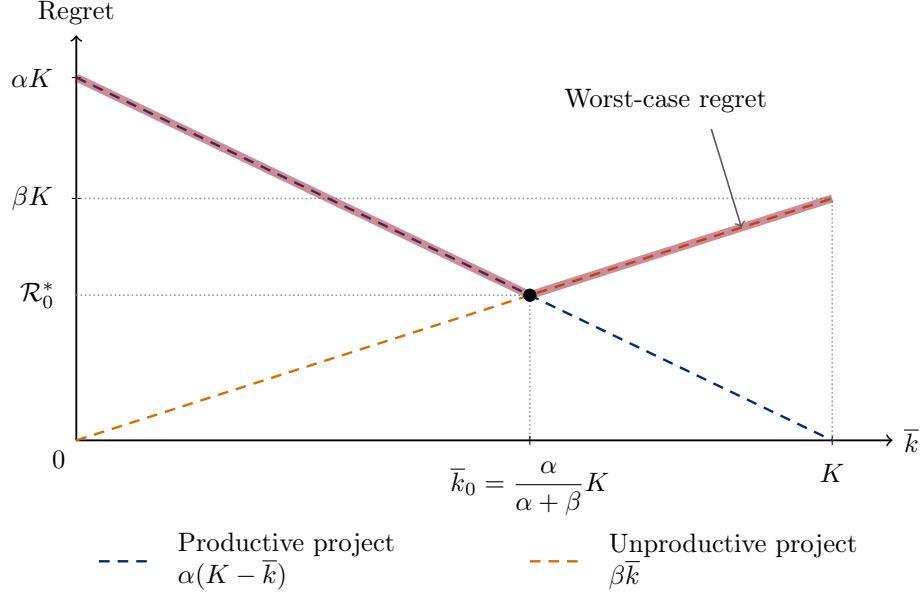

Proposition \ref{prop_delegation} implies that the optimal policy can be implemented by a simple spending rule: headquarters sets a ceiling, and the informed manager chooses how much to invest below it. This form resembles the initial spending limit in \citet{HarrisRaviv1996}, but in their model, a request above the initial limit may trigger an audit, after which additional funding can be granted. By contrast, our cap is a hard constraint: headquarters fixes a spending cap for the relevant planning period, and the manager cannot obtain additional capital. This cap is set to minimize the largest possible regret when headquarters lacks a prior. 

Having a hard spending limit is not uncommon in practice. \citet[p.~2032]{Graham2022} observes that corporate budgets are ``difficult to change during the budget year.'' \citet{DecaireSosyura2023} provide more direct evidence: managers who exhaust their budgets stop spending until the next budget reset, regardless of the investment opportunities that subsequently arise. 

We conclude this subsection by pointing out that the spending limit is not necessarily the unique least-regret capital rationing. For example, the binary menu \(\{0,\capex_0\}\) attains the same worst-case regret when $B \geq 1$. Nevertheless, the interval menu is practically natural and simple for implementation. 

\subsection{Comparative statics}
This subsection studies how the optimal spending limit and the associated regret vary with the primitives.

\begin{corollary}\label{cor:no-compensation-statics}
The optimal spending limit satisfies
\[
\frac{\partial\capex_0}{\partial A}>0,
\qquad
\frac{\partial\capex_0}{\partial K}>0,
\qquad
\frac{\partial\capex_0}{\partial r}<0,
\qquad
\frac{\partial\capex_0}{\partial B}
\begin{cases}
<0,& B<1,\\
=0,& B>1.
\end{cases}
\]
The minimized regret satisfies
\[
\frac{\partial\regret_0^*}{\partial A}>0,
\qquad
\frac{\partial\regret_0^*}{\partial K}>0,
\qquad
\frac{\partial\regret_0^*}{\partial r}<0,
\qquad
\frac{\partial\regret_0^*}{\partial B}
\begin{cases}
>0,& B<1,\\
=0,& B>1.
\end{cases}
\]
The spending limit and minimized regret are continuous at $B=1$, and $\capex_0/K$ is independent of $K$.
\end{corollary}

The effects of productivity bound $A$ and private-benefit bound $B$ are fairly straightforward. If the investment is potentially more productive, a restrictive policy can be more costly. Therefore, the spending limit gets higher. By contrast, higher potential private benefits make a permissive policy more costly, so the headquarters should reduce the spending limit. Once $B$ reaches one, however, the worst overinvestment case already involves a project with no output. Further increases in private benefits do not worsen this case, so the optimal limit no longer changes with $B$. A higher cost of capital $r$ lowers the present value of future output. Restricting investment is then less costly, so headquarters sets a lower limit.

The results on regret are related. Regret increases with $A$ because restricting a productive project becomes more costly. It also increases with $B$ until $B$ reaches one because permitting investment becomes more costly. As $B$ approaches zero, the manager becomes aligned with headquarters, the spending limit approaches $K$, and regret vanishes. As $A$ approaches $1+r$ from above, little productive surplus can be lost, so the spending limit approaches zero.

These results provide several empirical predictions. Divisions with greater productive upside should receive higher spending limits, while those with greater potential private benefits should receive lower limits. Spending limits should also be lower when the cost of capital is higher. Note that these predictions concern permitted spending rather than realized investment, which also depends on the manager's private information.

\section{Least Regret Contract Design}\label{sec:joint}

\subsection{The Role of Compensation}\label{sec:role-compensation}

We now bring compensation back into the problem. The central question is when and how compensation can improve on a spending cap. An increasing payment schedule encourages investment and reinforces the manager's private benefits. On the other hand, a declining schedule rewards lower investment and counteracts the manager’s incentives to expand, justifying a higher spending limit.  
Meanwhile, compensation is costly and may also discourage investment in a productive project.

Under a decreasing compensation schedule, the same investment may reflect underinvestment in a productive project or overinvestment in an unproductive one. In the first case, the manager stops investing because further investment reduces her compensation. In the second, she invests because her private benefits outweigh the reduction in compensation. Because headquarters cannot distinguish these cases from the investment choice alone, the contract must limit regret in both.

Let us see this problem more formally. Consider an interval menu and a differentiable, decreasing, concave payment schedule, and suppose the manager chooses an interior investment\footnote{
That is, $k_A\in \text{int}\ \menu$ is not driven by capital rationing.
}
$k_A$ that satisfies
\begin{equation}
    \max\{0,B-1\}\leq-w'(k_A)\leq\alpha
\qquad\text{and}\qquad
-w'(k_A)<B.\label{slope-restriction}
\end{equation}
We show that this same investment can result from the two projects below. 

\paragraph{A productive project.} Consider first a productive project for which informed headquarters prefers full investment. Let
\[
q_H(z)=(1+r)(1-w'(k_A))z,
\qquad b_H(z)=0.
\]
The manager's payoff from investing $z$ is
\[
\frac{q_H(z)}{1+r}-z+b_H(z)+w(z)
=\bigl(1-w'(k_A)\bigr)z-z+w(z)
=-w'(k_A)z+w(z).
\]
Thus, the manager chooses $k_A$.\footnote{At $z=k_A$, the derivative of the manager's payoff is $-w'(k_A)+w'(k_A)=0$. Because the payoff is concave, $k_A$ is a best response.} Headquarters, however, would invest $K$. Its regret is therefore
\[
w(k_A)-w'(k_A)(K-k_A).
\]

\paragraph{An unproductive project.} Consider next an unproductive project for which informed headquarters prefers no investment. Let
\[
q_L(z)=(1+r)(1-B-w'(k_A))z,
\qquad b_L(z)=Bz.
\]
The manager's payoff from investing $z$ is
\[
\frac{q_L(z)}{1+r}-z+b_L(z)+w(z)
=\bigl(1-B-w'(k_A)\bigr)z-z+Bz+w(z)
=-w'(k_A)z+w(z).
\]
Thus, the manager again chooses $k_A$. Headquarters' operating profit is $-(B+w'(k_A))z<0$, so it would invest zero. Its regret is therefore
\[
w(k_A)+(B+w'(k_A))k_A.
\]

Combining the two cases, the same investment $k_A$ produces the following underinvestment and overinvestment regrets:
\begin{equation}\label{eq:paired-regrets}
	\underbrace{w(k_A)-w'(k_A)(K-k_A)}_{\text{underinvestment}}
	\quad\text{and}\quad
	\underbrace{w(k_A)+(B+w'(k_A))k_A}_{\text{overinvestment}}.
\end{equation}
The payment schedule must control both.

\subsection{Optimal policy}

We first ask when compensation can improve on the spending cap in Proposition \ref{prop_delegation}.

\paragraph{When compensation cannot help.} We first show that compensation cannot help when $B\geq\min\{\hat A,2\}$, in which case the optimal policy remains the spending cap from Section \ref{sec:control}.
Specifically, there are two cases in which compensation may be ineffective. The first is when $B\geq\hat A$. Compare a productive project with $q(k)=Ak$ and $b(k)=0$ with an unproductive project with $q(k)=0$ and $b(k)=\hat A k$. In both cases, the manager receives $\alpha k+w(k)$ and therefore chooses the same investment $k_A$, although informed headquarters would choose $K$ for the productive project and zero for the unproductive one. The two regrets are $\alpha(K-k_A)+w(k_A)$ and $k_A+w(k_A)$. Note that the term $w(k_A)$ enters both in a positive manner, so their maximum cannot fall below $\frac{\alpha}{\alpha + 1}K$, which is already attained without compensation.

The second case occurs for $B\geq2$. Consider an unproductive project with $q(k)=0$ and $b(k)=Bk$. Suppose the largest permitted investment is $\capex$ and compensation induces $k_A<\capex$. For the manager to prefer $k_A$ to $\capex$, we have
\[
(B - 1)k_A+w(k_A)\geq(B-1)\capex+w(\capex)
\quad\Rightarrow\quad
w(k_A)\geq(B-1)(\capex-k_A),
\]
where we use $w(\capex) \geq 0$ because of limited liability.
Headquarters' regret is therefore at least
\[
k_A+w(k_A)\geq \capex + \underbrace{(B-2)}_{\geq  0}(\capex - k_A)\geq\capex.
\]
Thus, persuading a manager with strong private benefits to invest below the cap requires compensation at least as large as the resulting reduction in the firm's overinvestment loss. For any spending cap $\capex$, worst-case regret is therefore at least $\max\{\capex,\alpha(K-\capex)\}$. Minimizing this bound over $\capex$ gives $\frac{\alpha}{\alpha + 1}K$, which is attained by the optimal spending cap without compensation.

\paragraph{Optimal payment schedule when compensation is useful.} For $B<\min\{\hat A,2\}$, productive investment and empire building are no longer completely indistinguishable. Compensation can reduce worst-case regret from investment choices within the permitted range. The cap limits the manager’s maximum investment but may prevent profitable expansion beyond it.

As discussed above, a useful compensation schedule should decline with investment. Such a schedule allows a manager with a weak project to retain more compensation by investing less, while a manager with a productive project may forgo compensation to invest more. Our first result is that compensation should equal zero at the cap. Any positive payment at the cap could be subtracted from every payment without changing the manager’s investment choice, while lowering headquarters’ regret.

Three sources of worst-case regret arise from any contract: underinvestment within the permitted range, overinvestment within that range, and profitable investment excluded by the cap. The two projects in Subsection \ref{sec:role-compensation} give the first two losses in expressions \eqref{eq:paired-regrets}, while a productive project that would invest beyond the cap gives the third. 

To illustrate the construction, first consider a decreasing, concave
payment schedule that satisfies the slope restrictions \eqref{slope-restriction} throughout the permitted investment range. Each investment can then arise under either project,
generating the two regrets described above. Together with the loss from
profitable investment excluded by the cap, these examples imply
\[
\begin{split}
	\wcr([0,\capex],w)\geq\max\Bigl\{&\max_{k\in[0,\capex]}\bigl[w(k)-w'(k)(K-k)\bigr],\\
	&\max_{k\in[0,\capex]}\bigl[w(k)+(B+w'(k))k\bigr],\\
	&\alpha(K-\capex)\Bigr\}.
\end{split}
\]
This illustrative case highlights the forces that guide the construction.
A steeper payment schedule offsets private benefits but may discourage
productive investment, while a tighter cap limits excessive spending
but may exclude profitable expansion. Headquarters seeks to balance
these forces. The argument below uses this balance to motivate a
candidate contract; Appendix \ref{app:prop2} establishes optimality over the full
contract class.

Within this illustrative class, these three bounds cover every admissible project and are attained by the corresponding linear projects.
Consider first a managerial choice strictly below the spending cap. 
Managerial optimality and nondecreasing private benefits imply that the marginal operating return cannot exceed the compensation forgone, \(-w'(k)\).
If headquarters prefers more investment, diminishing returns therefore bound the remaining productive surplus by $-w'(k)(K-k)$.

If headquarters instead prefers less investment, the manager's willingness to invest up to $k$ implies that the marginal operating loss cannot exceed $B+w'(k)$. Diminishing returns then bound the total overinvestment loss by $(B+w'(k))k$. This bound also applies when the manager chooses the spending cap. 
The two linear projects in Subsection \ref{sec:role-compensation} attain these underinvestment and overinvestment bounds, respectively.

When the spending cap binds, the manager may still have a positive marginal incentive to expand but is prevented from doing so. If headquarters also prefers more investment, the bound on marginal productivity limits the forgone operating profit to $\alpha(K-\capex)$. Because compensation is zero at the cap, this is also a bound on regret. The project with $q(k)=Ak$ and $b(k)=0$ attains this loss: the manager invests up to the cap, while informed headquarters would invest the entire budget.

We can now construct the optimal policy. For a fixed cap, choose $w$ to keep the worst cases of the two within-cap losses as low as possible.
Fix a candidate worst-case value $\regret$. 
The examples suggest constructing a candidate schedule by keeping
the relevant regret equal to $\regret$ over different parts of the
permitted range.
Where the underinvestment bound binds, $w^*$ satisfies
\begin{equation*}\label{eq:equal-under}
	w^*(k)-w^{*\prime}(k)(K-k)\equiv\regret.
\end{equation*}
Where the overinvestment bound binds, it satisfies
\begin{equation*}\label{eq:equal-over}
	w^*(k)+(B+w^{*\prime}(k))k\equiv\regret.
\end{equation*}
Together with $w^*(\capex)=0$, these equations determine the payment schedule for a given cap. Headquarters then chooses the cap so that the regret from excluded productive investment, $\alpha(K-\capex)$, equals the worst within-cap regret. The resulting schedule takes one of two forms: either the underinvestment bound determines compensation throughout the permitted range, or it determines compensation below a cutoff $t$ and the overinvestment bound determines compensation above it. This construction yields the paper's main result.

Let $\regret^*:=\min_{\menu,w}\wcr(\menu,w)$ denote the minimum worst-case regret over all contracts.

\begin{proposition}\label{prop_contract_k}
	An optimal contract has menu $\menu^*=[0,\capex^*]$ and takes one of the following forms.
	\begin{enumerate}[label=(\roman*)]
		\item If $B\geq\min\{\hat A,2\}$, then $\capex^*=\frac{\alpha}{\alpha + 1}K$, $w^*\equiv0$, and $\regret^*=\frac{\alpha}{\alpha + 1}K$.
		
		\item If $0<\alpha<\frac{B}{3}$ and $B<\min\{\hat A,2\}$, then $\capex^*=\frac{\alpha}{B}K$, $\regret^*=\alpha K(1-\frac{\alpha}{B})$, and $w^*(k)=\regret^*\, \frac{\capex^*-k}{K-k}$ for $k\in[0,\capex^*]$.
		
		\item If $\alpha\geq \frac{B}{3}$ and $B<\min\{\hat A,2\}$, let
		\begin{equation}\label{eq:konly-parameters}
			t=\begin{cases}\frac{K}{3},&0<B\leq\frac{3}{2},\\\frac{B-1}{B}K,&\frac{3}{2}<B<2.\end{cases}
		\end{equation}
		Define $a=\frac{B}{2}t^2(1-\frac{2t}{K})$ and $d=\frac{Bt(K-t)^2}{K}$. The optimal cap is the unique $\capex^*\in[t,K)$ satisfying
		\begin{equation}\label{eq:konly-crossing}
			\alpha(K-\capex^*)=\frac B2\capex^*+\frac a{\capex^*}.
		\end{equation}
		Equivalently, $\capex^*=[\alpha K+\sqrt{\alpha^2K^2-2(2\alpha+B)a}]/(2\alpha+B)$. Moreover, $\regret^*=\alpha(K-\capex^*)$ and
		\begin{equation}\label{eq:konly-contract}
			w^*(k)=\begin{cases}
				\regret^*-\dfrac d{K-k},&0\leq k\leq t,\\[5pt]
				\regret^*-\dfrac B2k-\dfrac ak,&t\leq k\leq\capex^*.
			\end{cases}
		\end{equation}
	\end{enumerate}
	In all cases, $w^*(\capex^*)=0$. In the two compensation regions, $w^*$ is positive below the cap, strictly decreasing, and strictly concave on the permitted interval.
\end{proposition}
Figure \ref{fig:optimal-compensation} illustrates a typical optimal
compensation schedule in Proposition \ref{prop_contract_k}. Compensation is highest at zero
investment and falls to zero at the spending cap. Because the
schedule is strictly concave, each additional unit invested
entails a larger reduction in compensation. The manager can
choose any investment up to the cap; investment above it is
not permitted.

\begin{figure}[htbp]
\centering
\begingroup
\definecolor{policyblue}{RGB}{0,89,153}
\begin{tikzpicture}[
  x=10.8cm,y=22cm,font=\small,
  axis/.style={black!80,line width=0.6pt,->,>=stealth},
  payment/.style={draw=policyblue,line width=1.4pt}
]
 
  \def\policyCap{0.64760301386069}
  \def\policyRegret{0.35239698613931}
  \def\policyJoin{0.33333333333333}
  \def\policyInitial{0.20424883799116}

  \fill[black!6] (\policyCap,0) rectangle (1,0.225);
  \draw[black!40,densely dashed]
    (\policyCap,0) -- (\policyCap,0.225);
  \node[align=center,text=black!65] at (0.825,0.135)
    {Investment\\not permitted};

  \draw[axis] (0,0) -- (1.045,0);
  \draw[axis] (0,0) -- (0,0.24)
    node[above right] {Compensation $w^*(k)$};
  \node[below left] at (0,0) {$0$};
  \node[below=27pt] at (0.5,0) {Investment $k$};

  \foreach \y/\lab in
    {0.05/0.05,0.10/0.10,0.15/0.15,0.20/0.20}{
    \draw[black!70] (0,\y) -- ([xshift=-2pt]0,\y)
      node[left=3pt] {$\lab$};
  }
  \draw (1,0) -- ([yshift=-2pt]1,0)
    node[below=3pt] {$K=1$};
  \node[below=4pt] at (\policyCap,0)
    {$\capex^*\approx 0.65$};

  \draw[payment]
    plot[domain=0:\policyJoin,samples=100,variable=\x]
      (\x,{\policyRegret-4/(27*(1-\x))})
    plot[domain=\policyJoin:\policyCap,samples=120,variable=\x]
      (\x,{\policyRegret-\x/2-1/(54*\x)});

  \fill[policyblue] (0,\policyInitial) circle[radius=2pt];
  \node[above right,text=policyblue] at (0,\policyInitial)
    {$w^*(0)\approx 0.20$};
  \fill[policyblue] (\policyCap,0) circle[radius=2pt];
\end{tikzpicture}
\endgroup
\caption{Optimal compensation as a function of investment.\newline
{\footnotesize Parameters: $\hat A=2$, $B=1$, and $K=1$. These parameters fall into Case (iii) of Proposition \ref{prop_contract_k}.}}
\label{fig:optimal-compensation}
\end{figure}
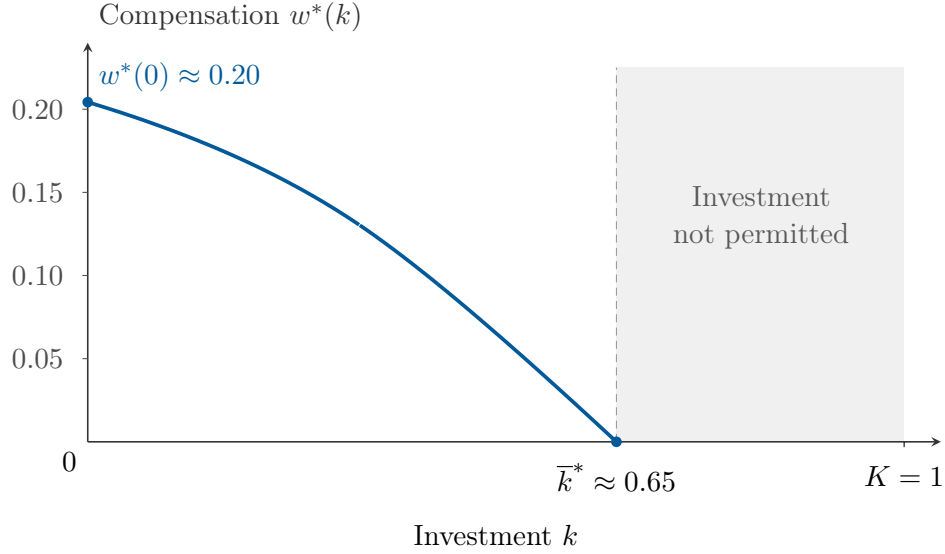

Proposition \ref{prop_contract_k} includes three cases. Case (i) considers situations in which potential private benefits are very high. Compensation is not useful because it cannot distinguish productive investment from empire building at a sufficiently low cost. The optimal contract therefore uses only a spending cap, with no compensation.

Case (ii) considers situations in which potential private benefits are moderate and potential productivity is relatively low. Compensation is useful because it can reduce investment in a weak project at a sufficiently low cost. The optimal contract therefore combines a declining compensation schedule with a relatively low spending cap. Compensation is chosen to control the regret from underinvestment throughout the permitted range, while the regret from overinvestment becomes equally important only at the cap. Figure \ref{fig:caseii-compensation-regret} illustrates this case.\footnote{
	For each target investment $k$, write $m(k)=\min\{-w'(k),\alpha\}$. The productive project $q(z)=(1+r)[1+m(k)]z$, $b(z)=[-w'(k)-m(k)]z$ is admissible for all three schedules in Figure \ref{fig:caseii-compensation-regret} and induces $k$, as does the unproductive project in the preceding construction. Its regret is $w(k)+m(k)(K-k)$. When $-w'(k)>\alpha$, positive private benefits sustain investment despite output already being at its upper bound.
} In panel (a), regret from the productive project equals $\regret^*$ throughout, while regret from the unproductive project reaches $\regret^*$ only at the cap. Panels (b) and (c) show that lowering or raising compensation increases the highest regret.

\begin{figure}[htbp]
	\centering
	\begin{tikzpicture}[x=21cm,y=22cm,font=\footnotesize]
	
		\def\caseTwoCap{0.2}
		\def\caseTwoRegret{0.16}
		\def\caseTwoD{0.128}
		\def\caseTwoAlpha{0.2}
		\def\caseTwoCrossLow{0.165374901231}
		\def\caseTwoSlopeLimit{0.160952921464}
		\tikzset{
			case two productive/.style={draw=pennblue,line width=1pt},
			case two unprofitable/.style={draw=orange!85!black,line width=1pt,dash pattern=on 4pt off 3pt},
			case two guide/.style={draw=black!40,densely dotted,line width=0.6pt},
			case two envelope/.style={draw=pennred,opacity=0.4,line width=3pt}
		}
	
		\draw[case two productive] (0.03,-0.082)--(0.048,-0.082);
		\node[anchor=west,align=left] at (0.053,-0.082)
		{Productive project\\$w(k)+\min\{-w'(k),\alpha\}(K-k)$};
		\draw[case two unprofitable] (0.377,-0.082)--(0.395,-0.082);
		\node[anchor=west,align=left] at (0.400,-0.082)
		{Unproductive project\\$w(k)+(B+w'(k))k$};
		\node[anchor=east] at (-0.007,\caseTwoRegret) {$\regret^*$};
		
		\foreach \caseTwoPanel/\caseTwoScale/\caseTwoShift in {0/1/0,1/0.9/5.2,2/1.1/10.4} {
			\begin{scope}[xshift=\caseTwoShift cm]
				\pgfmathsetmacro{\caseTwoScaledRegret}{\caseTwoScale*\caseTwoRegret}
				\pgfmathsetmacro{\caseTwoOrangeCap}{\caseTwoScaledRegret+\caseTwoCap-\caseTwoScale*\caseTwoD/((1-\caseTwoCap)^2)}
				\draw[->,line width=0.7pt] (0,0)--(0.215,0) node[right] {$k$};
				\draw[->,line width=0.7pt] (0,0)--(0,0.19) node[above] {Regret};
				\node[below left] at (0,0) {$0$};
				\draw[case two guide] (0,\caseTwoRegret)--(0.205,\caseTwoRegret);
				\draw (\caseTwoCap,0.002)--(\caseTwoCap,-0.002)
				node[below=4pt] {$\capex^*$};
				
				\ifcase\caseTwoPanel
				\draw[case two envelope] (0,\caseTwoRegret)--(\caseTwoCap,\caseTwoRegret);
				\or
				\draw[case two envelope] (0,\caseTwoScaledRegret)--(\caseTwoCrossLow,\caseTwoScaledRegret);
				\draw[case two envelope,smooth,samples=100,domain=\caseTwoCrossLow:\caseTwoCap]
				plot (\x,{\caseTwoScaledRegret+\x-\caseTwoScale*\caseTwoD/((1-\x)^2)});
				\or
				\draw[case two envelope] (0,\caseTwoScaledRegret)--(\caseTwoSlopeLimit,\caseTwoScaledRegret);
				\draw[case two envelope,smooth,samples=100,domain=\caseTwoSlopeLimit:\caseTwoCap]
				plot (\x,{\caseTwoScaledRegret-\caseTwoScale*\caseTwoD/(1-\x)+\caseTwoAlpha*(1-\x)});
				\fi
				
				\ifnum\caseTwoPanel=2
				\draw[case two productive] (0,\caseTwoScaledRegret)--(\caseTwoSlopeLimit,\caseTwoScaledRegret);
				\draw[case two productive,smooth,samples=100,domain=\caseTwoSlopeLimit:\caseTwoCap]
				plot (\x,{\caseTwoScaledRegret-\caseTwoScale*\caseTwoD/(1-\x)+\caseTwoAlpha*(1-\x)});
				\else
				\draw[case two productive] (0,\caseTwoScaledRegret)--(\caseTwoCap,\caseTwoScaledRegret);
				\fi
				\draw[case two unprofitable,smooth,samples=150,domain=0:\caseTwoCap]
				plot (\x,{\caseTwoScaledRegret+\x-\caseTwoScale*\caseTwoD/((1-\x)^2)});
				
				\ifcase\caseTwoPanel
				\node[anchor=south] at (0.1,0.217) {(a) $w=w^*$};
				\draw[case two guide] (\caseTwoCap,0)--(\caseTwoCap,\caseTwoRegret);
				\fill (\caseTwoCap,\caseTwoRegret) circle[radius=2pt];
				\or
				\node[anchor=south] at (0.1,0.217) {(b) $w=0.9w^*$};
				\draw[case two guide] (\caseTwoCrossLow,0)--(\caseTwoCrossLow,\caseTwoScaledRegret);
				\draw[case two guide] (\caseTwoCap,0)--(\caseTwoCap,\caseTwoOrangeCap);
				\fill (\caseTwoCrossLow,\caseTwoScaledRegret) circle[radius=2pt];
				\fill (\caseTwoCap,\caseTwoOrangeCap) circle[radius=2pt];
				\or
				\node[anchor=south] at (0.1,0.217) {(c) $w=1.1w^*$};
				\draw[case two guide] (\caseTwoCap,0)--(\caseTwoCap,\caseTwoRegret);
				\draw[case two guide] (\caseTwoSlopeLimit,0)--(\caseTwoSlopeLimit,\caseTwoScaledRegret);
				\draw (\caseTwoSlopeLimit,0.002)--(\caseTwoSlopeLimit,-0.002);
				\fill (\caseTwoSlopeLimit,\caseTwoScaledRegret) circle[radius=2pt];
				\fi
			\end{scope}
		}
	\end{tikzpicture}
	\caption{Polar regrets under optimal and nearby suboptimal compensation, Case (ii).
    \newline
		{\footnotesize Parameters: $\hat A=1.2$, $B=1$, $K=1$, and the same cap $\capex^*=0.2$.}}
	\label{fig:caseii-compensation-regret}
\end{figure}
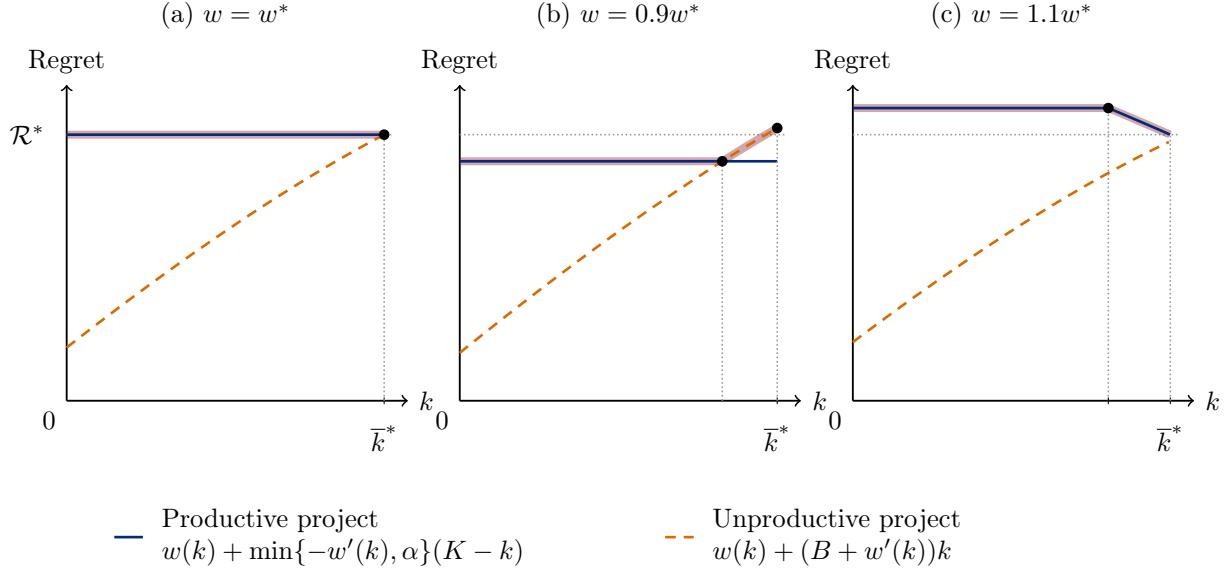

Finally, Case (iii) considers situations in which potential private benefits are moderate and potential productivity is relatively high. A restrictive spending cap is now more costly. Unlike Case (ii), neither underinvestment nor overinvestment alone determines compensation throughout the permitted range. Instead, there is a cutoff investment $t$ at which the two regrets are equal. Compensation is chosen to control the regret from underinvestment below $t$ and the regret from overinvestment above $t$. The regret from productive investment excluded by the cap also equals $\regret^*$. Figure \ref{fig:compensation-regret-tradeoff} illustrates Case (iii). In panel (a), regret from the productive project equals $\regret^*$ below $t$, while regret from the unproductive project equals $\regret^*$ above $t$. Panels (b) and (c) show that uniformly lowering or raising compensation causes one of these losses to exceed $\regret^*$.

\begin{figure}[htbp]
	\centering
	\begin{tikzpicture}[x=6.5cm,y=9cm,font=\footnotesize]
		\def\compFigA{2}
		\def\compFigB{1}
		\def\compFigBudget{1}
		\pgfmathsetmacro{\compFigAlpha}{\compFigA-1}
		\pgfmathsetmacro{\compFigJoin}{\compFigBudget/3}
		\pgfmathsetmacro{\compFigLittleA}{\compFigB*\compFigJoin^2*(1-2*\compFigJoin/\compFigBudget)/2}
		\pgfmathsetmacro{\compFigD}{\compFigB*\compFigJoin*(\compFigBudget-\compFigJoin)^2/\compFigBudget}
		\pgfmathsetmacro{\compFigCap}{(\compFigAlpha*\compFigBudget+sqrt(\compFigAlpha^2*\compFigBudget^2-2*(2*\compFigAlpha+\compFigB)*\compFigLittleA))/(2*\compFigAlpha+\compFigB)}
		\pgfmathsetmacro{\compFigRegret}{\compFigAlpha*(\compFigBudget-\compFigCap)}
		\def\compFigCrossLow{0.218078018146}
		\def\compFigCrossHigh{0.448921785294}
		\tikzset{
			comp productive/.style={draw=pennblue,line width=1pt},
			comp unprofitable/.style={draw=orange!85!black,line width=1pt,dash pattern=on 4pt off 3pt},
			comp guide/.style={draw=black!40,densely dotted,line width=0.6pt},
			comp envelope/.style={draw=pennred,opacity=0.4,line width=3pt}
		}
		
		\draw[comp productive] (0.17,-0.2)--(0.23,-0.2);
		
		\node at (-0.05, 0.35) {$\regret^*$};
		\node[anchor=west,align=left] at (0.25,-0.2)
		{Productive project\\$w(k)-w'(k)(K-k)$};
		\draw[comp unprofitable] (1.22,-0.2)--(1.28,-0.2);
		\node[anchor=west,align=left] at (1.30,-0.2)
		{Unproductive project\\$w(k)+(B+w'(k))k$};
		
		\foreach \compFigPanel/\compFigScale/\compFigXShift in {0/1/0,1/0.9/5.2,2/1.1/10.4} {
			\begin{scope}[xshift=\compFigXShift cm]
				\pgfmathsetmacro{\compFigScaledRegret}{\compFigScale*\compFigRegret}
				\pgfmathsetmacro{\compFigCapRegret}{\compFigScaledRegret+(1-\compFigScale)*\compFigB*\compFigCap}
				
				\draw[->,line width=0.7pt] (0,0)--({\compFigCap+0.05},0)
				node[right] {$k$};
				\draw[->,line width=0.7pt] (0,0)--(0,0.43) node[above] {Regret};
				\node[below left] at (0,0) {$0$};
				
				\draw[comp guide] (0,\compFigRegret)--({\compFigCap+0.018},\compFigRegret);
				\draw[comp guide] (\compFigCap,0)--(\compFigCap,\compFigCapRegret);
				\draw (\compFigCap,0.005)--(\compFigCap,-0.005)
				node[below=4pt,font=\footnotesize] {$\capex^*$};
				
				\ifcase\compFigPanel
				\draw[comp envelope] (0,\compFigRegret)--(\compFigCap,\compFigRegret);
				\or
				\draw[comp envelope] (0,\compFigScaledRegret)--(\compFigCrossLow,\compFigScaledRegret);
				\draw[comp envelope,smooth,samples=100,domain=\compFigCrossLow:\compFigJoin]
				plot (\x,{\compFigScaledRegret+\compFigB*\x-\compFigScale*\compFigD*\compFigBudget/((\compFigBudget-\x)^2)});
				\draw[comp envelope]
				(\compFigJoin,{\compFigScaledRegret+(1-\compFigScale)*\compFigB*\compFigJoin})
				--(\compFigCap,\compFigCapRegret);
				\or
				\draw[comp envelope] (0,\compFigScaledRegret)--(\compFigJoin,\compFigScaledRegret);
				\draw[comp envelope,smooth,samples=100,domain=\compFigJoin:\compFigCrossHigh]
				plot (\x,{\compFigScale*(\compFigRegret+\compFigB*\compFigBudget/2-\compFigB*\x-\compFigLittleA*\compFigBudget/(\x^2))});
				\draw[comp envelope]
				(\compFigCrossHigh,{\compFigScaledRegret+(1-\compFigScale)*\compFigB*\compFigCrossHigh})
				--(\compFigCap,\compFigCapRegret);
				\fi
				
				\draw[comp productive] (0,\compFigScaledRegret)--(\compFigJoin,\compFigScaledRegret);
				\draw[comp productive,smooth,samples=120,domain=\compFigJoin:\compFigCap]
				plot (\x,{\compFigScale*(\compFigRegret+\compFigB*\compFigBudget/2-\compFigB*\x-\compFigLittleA*\compFigBudget/(\x^2))});
				\draw[comp unprofitable,smooth,samples=120,domain=0:\compFigJoin]
				plot (\x,{\compFigScaledRegret+\compFigB*\x-\compFigScale*\compFigD*\compFigBudget/((\compFigBudget-\x)^2)});
				\draw[comp unprofitable]
				(\compFigJoin,{\compFigScaledRegret+(1-\compFigScale)*\compFigB*\compFigJoin})
				--(\compFigCap,\compFigCapRegret);
				
				\ifcase\compFigPanel
				\node[anchor=south] at (0.324,0.49) {(a) $w=w^*$};
				\draw[comp guide] (\compFigJoin,0)--(\compFigJoin,\compFigRegret);
				\draw (\compFigJoin,0.005)--(\compFigJoin,-0.005)
				node[below=4pt,font=\footnotesize] {$t$};
				\fill (\compFigJoin,\compFigRegret) circle[radius=2pt];
				\or
				\node[anchor=south] at (0.324,0.49) {(b) $w=0.9w^*$};
				\draw[comp guide] (\compFigCrossLow,0)--(\compFigCrossLow,\compFigScaledRegret);
				\fill (\compFigCrossLow,\compFigScaledRegret) circle[radius=2pt];
				\fill (\compFigCap,\compFigCapRegret) circle[radius=2pt];
				\or
				\node[anchor=south] at (0.324,0.49) {(c) $w=1.1w^*$};
				\pgfmathsetmacro{\compFigCrossValue}{\compFigScaledRegret+(1-\compFigScale)*\compFigB*\compFigCrossHigh}
				
				\draw[comp guide] (\compFigCrossHigh,0)--(\compFigCrossHigh,\compFigCrossValue);
			
				\fi
			\end{scope}
		}
	\end{tikzpicture}
	\caption{Polar regrets under optimal and nearby suboptimal compensation, Case (iii). {\footnotesize Parameters: $\hat A=2$, $B=1$, $K=1$, and the same cap $\capex^*\approx 0.65$.}}
	\label{fig:compensation-regret-tradeoff}
\end{figure}
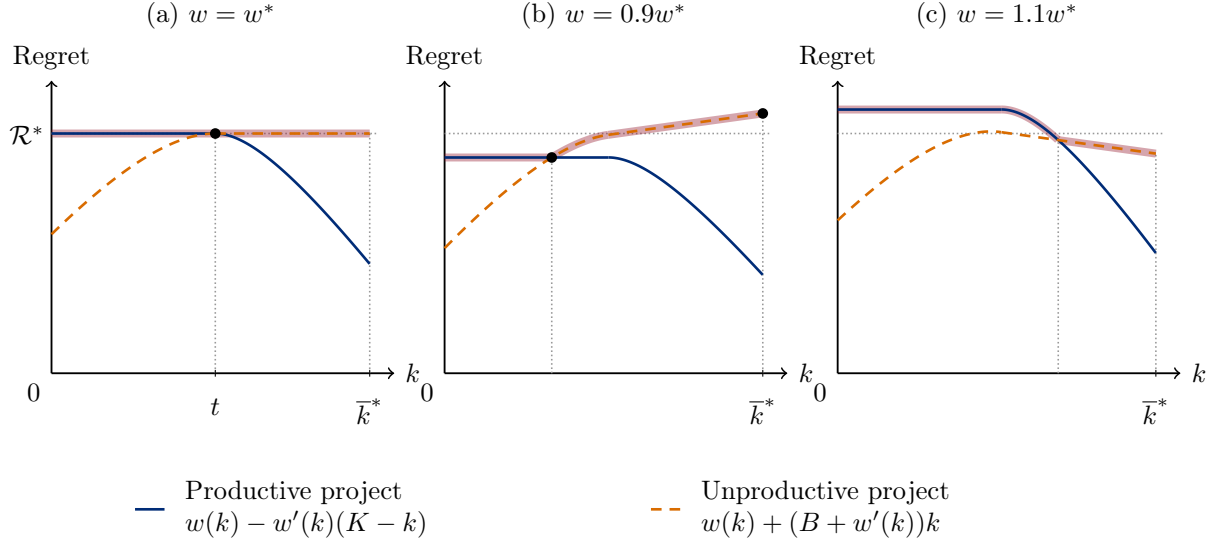

Before concluding this subsection, we compare the optimal contract with pure capital rationing in Section \ref{sec:control}. Let $\capex_0$ be the optimal cap without compensation and $\capex^*$ the optimal cap with compensation.

\begin{corollary}\label{cor:delegation}
	If $B<\min\{\hat A,2\}$, then
	\[
	\capex^*>\capex_0,
	\quad
	\regret^*<\regret_0^*.
	\]
	If $B\geq\min\{\hat A,2\}$, the two caps and regret values are identical.
\end{corollary}

Corollary \ref{cor:delegation} implies that compensation complements managerial discretion. When compensation is useful, rewarding lower investment limits spending on poor projects and allows headquarters to set a higher cap, which preserves more of the upside from productive projects and lowers worst-case regret. Importantly, a higher cap expands permitted investment but does not increase investment for every project. A manager with a poor project may invest less in response to compensation, while one with a productive project may use the higher cap and receive little or no compensation. At the same time, limited liability explains why compensation does not eliminate the spending cap. Headquarters cannot charge the manager for investing more; it can only reward lower investment and withdraw that reward as investment increases. Without a cap, headquarters would have to offer large payments for restraint to offset the manager's private benefit from ever-higher investment. The optimal contract therefore combines a declining payment schedule with a spending cap.

\subsection{Comparative statics of the optimal contract}
\label{sec:comparative-statics}

We next study how the spending cap and compensation vary with the primitives. A change in $A$ or $B$ changes the range of productive returns or private benefits that headquarters considers possible, rather than the characteristics of a particular project or manager.

\begin{corollary}\label{cor:cap-statics}
	The optimal spending cap increases with $A$ and $K$ and decreases with $r$. It decreases with $B$ when compensation is useful and is independent of $B$ when compensation is not useful. 
\end{corollary}

The results in Corollary \ref{cor:cap-statics} are largely similar to those in Corollary \ref{cor:no-compensation-statics}. The main difference concerns potential private benefits $B$. Without compensation, the spending cap no longer changes with $B$ once $B$ reaches one. With compensation, a higher $B$ continues to lower the cap as long as compensation is useful because it also raises the cost of inducing lower investment. As a result, the cap becomes independent of $B$ only when compensation is no longer useful.

\begin{corollary}\label{cor:compensation-statics}
	When compensation is useful, holding investment fixed, higher $A$ raises compensation, while higher $r$ lowers it within each compensation region. Higher $B$ makes compensation decline more steeply with investment but has an ambiguous effect on its level. 
\end{corollary}
The results on compensation are fairly intuitive: a higher $A$ raises compensation, while a higher $r$ lowers it. Moreover, a higher $B$ makes compensation decline more steeply with investment because headquarters must counter stronger private benefits from project scale. Interestingly, the effect on the level of compensation is ambiguous: payments may fall at every investment level or rise at low levels of investment and fall at high levels of investment. 

These results have several testable implications. Business units with greater potential productive upside should receive higher spending caps and higher compensation schedules. Greater managerial private benefits should be associated with lower caps and compensation that is more sensitive to investment, but not necessarily with lower payments. A higher cost of capital should be associated with both lower spending caps and lower compensation. 

\subsection{Discussion and extensions}\label{sec:extensions}
Our analysis restricts compensation to depend on investment, as may be appropriate when output is difficult to verify or arrives too late to determine the manager's pay. If headquarters can instead condition compensation on both investment and realized output, it can make payments for restraint depend on whether the investment covers its cost. This relaxes the informational constraint in Section \ref{sec:role-compensation}: the same investment can produce different outputs in a productive project and one sustained by private benefits. Observing output nevertheless reveals only the return at the chosen investment, not the entire production function.

Appendix \ref{app:outcome-compensation} studies this possibility. The construction combines a spending cap with payments for restraint, offered only when realized operating profit is \textit{nonpositive}. 
Within the payment region, the manager forgoes compensation when she expands investment. 
At a positive investment \(k>0\), nonpositive operating profit, \(q(k)/(1+r)\leq k\), and diminishing returns rule out profitable expansion beyond that investment.
Headquarters can therefore reward restraint in this case without worrying too much about sacrificing profit from further expansion. 
On the other hand, when an investment is profitable, headquarters need not provide an additional incentive to invest: the manager already benefits from operating profit and enjoys private benefits from investment. A payment for restraint could instead discourage profitable expansion. The contract therefore withholds compensation in this case and relies on the spending cap to limit excessive investment.

Under a parameter restriction, contracts of this form can strictly improve on the optimal investment-only contract and approach a sharp lower bound on regret.\footnote{
For $0<B<2$ and $0<\alpha\leq \frac{B}{2-B}$, where $\alpha= \frac{A}{1+r}-1$, the infimum over the output-contingent contracts defined in the appendix is $\frac{\alpha B}{2\alpha + B}K$.  For $0<B<2$ outside this parameter region, whether the same bound is sharp remains open.
} The usefulness of compensation is thus no longer limited by the condition $B<\hat A$.  Output contingency expands the scope for compensation.
Nevertheless, its value is limited. Relative to the optimal investment-only contract, output contingency reduces worst-case regret by at most one ninth. When potential private benefits satisfy \(B\geq2\), it provides no improvement.
Costly auditing would additionally require headquarters to decide when the value of conditioning payments on output justifies the cost of verification, which narrows the benefit even further.

\section{Conclusion}

This paper studies capital budgeting when headquarters cannot assign a prior to either investment opportunities or managerial private benefits. Evaluating policies by the largest profit shortfall relative to informed investment yields a simple capital rationing policy: when compensation is unavailable, a spending cap is optimal even among procedures that can exclude particular investment levels. The cap balances losses from empire building against losses from restricting productive projects. It limits the manager's authority to expand while preserving her discretion over how much of the allocation to use.

Compensation can expand this discretion by giving the manager a reason to preserve capital. When potential private benefits are sufficiently limited, an optimal contract combines a higher cap with a decreasing, strictly concave payment schedule that reaches zero at the cap. Expanding the project requires the manager to forgo compensation at an increasing marginal rate. Paying for restraint thus supports greater investment authority while limiting the losses from excessive spending. Compensation ceases to improve the guarantee when potential private benefits are large enough to make productive investment and empire building indistinguishable, or to make restraint at least as costly as the investment loss it avoids. A spending cap alone is then optimal even though payments are available.

The analysis highlights the value of designing investment authority and compensation jointly. Greater potential productivity raises the cost of restricting investment and warrants a higher cap. When compensation remains useful, greater potential private benefits call for a lower cap and a steeper decline in payments, but need not raise the payment offered at a given investment. These implications concern headquarters' information about possible projects and managerial preferences. They suggest examining spending limits and the sensitivity of compensation to investment together in robust capital budgeting design. 

\bibliographystyle{ecta}
\bibliography{references}

\newpage
\appendix
\section{Proof of Lemma \ref{lem_direct}}\label{app:lemma}
Fix a regular mechanism $(M,\kappa,\omega)$, and let
\[
\mathcal O:=\{(\kappa(m),\omega(m)):m\in M\}.
\]
By admissibility, $\mathcal O$ is nonempty, compact, and contains $\mathbf{0}$. Define
\[
\menu:=\{k:(k,z)\in\mathcal O\text{ for some }z\},
\qquad
w(k):=\max\{z:(k,z)\in\mathcal O\},\quad k\in\menu.
\]
The menu is nonempty and compact, and each maximum exists. Moreover, $w$ is finite, nonnegative, and upper semicontinuous. To see the last assertion, take $k_n\to k$ in $\menu$ and a subsequence along which $w(k_n)$ converges to its limsup. Compactness of $\mathcal O$ gives
\[
\bigl(k,\limsup_n w(k_n)\bigr)\in\mathcal O,
\qquad
\limsup_n w(k_n)\leq w(k).
\]
Extending $w$ by zero outside the closed menu preserves upper semicontinuity because $w\geq0$.

For any environment, write $v(k):=\frac{q(k)}{1+r}-k+b(k)$. An outcome $(k,z)\in\mathcal O$ with $z<w(k)$ is strictly worse for the manager than $(k,w(k))$. Consequently,
\[
\operatorname*{arg\,max}_{(k,z)\in\mathcal O}\{v(k)+z\}
=
\left\{(k,w(k)):k\in
\operatorname*{arg\,max}_{x\in\menu}\{v(x)+w(x)\}\right\}.
\]
Choosing a message in $M$ is equivalent to choosing its outcome in $\mathcal O$, so this equality identifies all best-response outcomes under the mechanism and the contract. At any such outcome $(k,w(k))$, the payoff pair for headquarters and the manager is
\[
\bigl(\pi(k;q)-w(k),\ \pi(k;q)+b(k)+w(k)\bigr)
\]
under either representation. Hence the sets of both parties' best-response payoff pairs coincide in every environment, including all ties.

Conversely, consider a nonempty closed menu $\menu \ni 0$ and a finite, nonnegative, upper semicontinuous schedule $w$. Set
\[
M_w:=\{(k,z):k\in\menu,\ 0\leq z\leq w(k)\},
\qquad \kappa(k,z):=k,\qquad \omega(k,z):=z.
\]
Since $w$ is finite and upper semicontinuous on the compact menu, it is bounded above. Upper semicontinuity also makes $M_w$ closed. Thus $M_w$, which is the outcome set of this mechanism, is nonempty and compact, so the mechanism is regular. At any investment $k$, the manager strictly prefers $w(k)$ to any smaller available payment. The mechanism therefore induces exactly the contract's best-response investment--payment pairs and, by the displayed payoff formula, exactly the same payoff pairs for both parties. Finally, if $\omega\equiv0$, the first construction gives $w\equiv0$; if $w\equiv0$, the converse gives $\omega\equiv0$. This establishes the equivalence also when compensation is unavailable.
\hfill\qedsymbol

\section{Proof of Proposition \ref{prop_delegation}}\label{app:prop1}
Write $\alpha:=\hat A-1$ and $\beta:=\min\{B,1\}$. Fix any nonempty closed menu $\menu$, and let $\capex:=\max\menu$. In the environment $(q(k),b(k))=(Ak,0)$, the manager chooses $\capex$, whereas the first-best investment is $K$. Hence worst-case regret is at least $\alpha(K-\capex)$. For the opposite concern, take
\[
(q(k),b(k))=\bigl((1+r)(1-\beta)k,\beta k\bigr).
\]
This environment is admissible, and the manager is indifferent among all investments. Adverse tie-breaking selects $\capex$, yielding regret $\beta\capex$. Thus
\begin{equation}
\wcr(\menu,0)\geq
\max\{\alpha(K-\capex),\beta\capex\}.
\label{eq:delegation-lower-bound}
\end{equation}

We next show that the interval $[0,\capex]$ attains this bound. Fix an environment, a managerial best response $k\in[0,\capex]$, and a first-best investment $p\in[0,K]$. If $p>\capex$, comparison with the feasible choice $\capex$ and monotonicity of $b$ imply
\[
\pi(k)\geq\pi(\capex)+b(\capex)-b(k)\geq\pi(\capex).
\]
Since $\hat q:= \frac{q}{1+r}$ is $\hat A$-Lipschitz, regret is then at most
\[
\pi(p)-\pi(k)\leq\alpha(p-\capex)
\leq\alpha(K-\capex).
\]
If $p\leq\capex$, the manager's preference for $k$ gives
\[
\pi(p)-\pi(k)\leq b(k)-b(p).
\]
When $k\leq p$, the right-hand side is nonpositive, so regret is zero. When $k>p$, the $B$-Lipschitz property of $b$ bounds regret by $B(k-p)$, while monotonicity of $q$ bounds it by $k-p$. Regret is therefore at most $\beta(k-p)\leq\beta\capex$. These inequalities hold for every best response, proving
\[
\wcr([0,\capex],0)
=\max\{\alpha(K-\capex),\beta\capex\}.
\]
Balancing the two terms gives $\capex_0=\alpha K/(\alpha+\beta)$ and $\regret_0^*=\beta\capex_0$. The lower bound \eqref{eq:delegation-lower-bound} establishes optimality among all closed menus.
\hfill\qedsymbol

\section{Proof of Proposition \ref{prop_contract_k}}
\label{apx_proof_contract_k}\label{app:prop2}

Write $\alpha:=\hat A-1$. Throughout the lower-bound argument, let
$(\menu,w)$ be any admissible contract, where $\menu$ is nonempty
and closed and $w$ is nonnegative and upper semicontinuous, and
let $\capex:=\max\menu$. The environment $(q(k),b(k))=(Ak,0)$ gives
the useful bound
\begin{equation}
\label{eq:konly-cap-lower}
\wcr(\menu,w)\geq \alpha(K-\capex),
\end{equation}
since every available investment is at most $\capex$ and payments
are nonnegative. None of the lower bounds below requires an
interval menu or differentiability of the competing schedule.

\medskip
\noindent\textbf{Step 1: Pure delegation.}
If $B\geq\hat A$, the environments $(Ak,0)$ and $(0,\hat A k)$ give the
agent the same objective $\alpha k+w(k)$. At any common best
response $z$, their regrets imply
\[
\wcr(\menu,w)
\geq\max\{\alpha(K-z),z\}+w(z)
\geq\frac{\alpha K}{\hat A}.
\]
If instead $B\geq2$, use $(q(k),b(k))=(0,Bk)$. A best response $z$
satisfies
\[
(B-1)z+w(z)\geq(B-1)\capex+w(\capex),
\]
so its regret is at least
\[
z+w(z)\geq(B-1)\capex-(B-2)z\geq\capex.
\]
Together with \eqref{eq:konly-cap-lower}, this gives the same lower
bound $\alpha K/\hat A$. Proposition \ref{prop_delegation} shows that
zero compensation and cap $\alpha K/\hat A$ attain it. These two
arguments cover $B\geq\min\{\hat A,2\}$.

\medskip
\noindent\textbf{Step 2: The lower bound in part \textup{(ii)}.}
Suppose $0<\alpha<B/3$ and $B<\hat A$, which also imply $B<3/2$.
The admissible environments
$(Ak,0)$ and $((1+r)(\hat A-B)k,Bk)$ again give the same objective
$\alpha k+w(k)$. Since $B-\alpha>0$, any common best response $z$
implies
\[
\begin{split}
\wcr(\menu,w)
&\geq\max\{\alpha(K-z),(B-\alpha)z\}+w(z)\\
&\geq\min_{x\in[0,K]}
       \max\{\alpha(K-x),(B-\alpha)x\}
=\alpha K\left(1-\frac{\alpha}{B}\right).
\end{split}
\]
This is the proposed value. The minimizing investment in the
last expression is $\capex^*=\alpha K/B$; this observation does
not impose the same maximum investment on every optimal menu.

\medskip
\begin{samepage}
\noindent\textbf{Step 3: The lower bound in part \textup{(iii)}.}
Suppose $\alpha\geq B/3$ and $B<\min\{\hat A,2\}$. Equivalently to
\eqref{eq:konly-parameters}, write
\[
s_0:=\max\{B/3,B-1\},\qquad
t=\frac{Ks_0}{B},\qquad
a=(B/2-s_0)t^2,\qquad d=s_0(K-t)^2.
\]
\end{samepage}
Then $B-1\leq s_0\leq\alpha$, $0<s_0<B/2$, and $0<t<K/2$.
For any $\capex\geq t$, define
\[
\rho_{\capex}:=\frac B2\capex+\frac a{\capex},\qquad
w_{\capex}(k):=
\begin{cases}
\rho_{\capex}-d/(K-k),&0\leq k\leq t,\\
\rho_{\capex}-Bk/2-a/k,&t<k\leq\capex.
\end{cases}
\]
The identities $s_0(K-t)=(B-s_0)t$ and
$a/t^2=B/2-s_0$ show that the branches agree in value and
derivative. Thus $w_{\capex}$ is continuously differentiable,
strictly decreasing and strictly concave, with
$w_{\capex}(\capex)=0$.

Every $x\in[0,\capex]$ is supported by an admissible environment
that gives regret $\rho_{\capex}$ under this reference schedule.
To see this, let
\[
s(x):=-w_{\capex}'(x)=
\begin{cases}d/(K-x)^2,&x\leq t,\\
B/2-a/x^2,&x>t.
\end{cases}
\]
Concavity makes $x$ a maximizer of $s(x)k+w_{\capex}(k)$.
If $x\leq t$, choose $(q(k),b(k))=((1+r)(1+s(x))k,0)$, which is
admissible because $0<s(x)\leq s_0\leq\alpha$. Its regret at $x$ is
$s(x)(K-x)+w_{\capex}(x)=\rho_{\capex}$.
If $x>t$, choose $((1+r)(1+s(x)-B)k,Bk)$, which is admissible because
$B-1\leq s_0<s(x)<B/2<1$. Its regret is
$(B-s(x))x+w_{\capex}(x)=\rho_{\capex}$.

Now choose
\[
x^*\in\arg\max_{k\in\menu}\{w(k)-w_{\capex}(k)\},\qquad
\Delta:=w(x^*)-w_{\capex}(x^*).
\]
Such a maximizer exists by upper semicontinuity and compactness,
and $\Delta\geq w(\capex)\geq0$. In the environment supporting $x^*$,
every $k\in\menu$ satisfies
\[
s(x^*)k+w(k)
\leq s(x^*)k+w_{\capex}(k)+\Delta
\leq s(x^*)x^*+w(x^*).
\]
Hence $x^*$ is also a best response to the actual contract,
and its regret is $\rho_{\capex}+\Delta$. Adverse tie-breaking yields
\begin{equation}
\label{eq:konly-reference-lower}
\wcr(\menu,w)\geq
\max\{\alpha(K-\capex),\rho_{\capex}\},\qquad \capex\geq t.
\end{equation}
The first term decreases strictly with $\capex$, whereas
$\rho_{\capex}'=B/2-a/\capex^2\geq s_0>0$.
Moreover, $\rho_t=s_0(K-t)\leq\alpha(K-t)$ and $\rho_K>0$.
Their unique intersection is the proposed cap $\capex^*\in[t,K)$.
For $\capex<t$, \eqref{eq:konly-cap-lower} exceeds
$\alpha(K-t)\geq\alpha(K-\capex^*)$.
Consequently every contract has regret at least
\begin{equation}
\label{eq:konly-candidate-value}
\regret^*=\alpha(K-\capex^*).
\end{equation}

\medskip
\noindent\textbf{Step 4: The proposed compensation attains the bound.}
It remains to verify the candidate in both compensation regions.
In part \textup{(ii)}, set $t=\capex^*$ and
$d=\regret^*(K-\capex^*)$; then
\[
w^*(k)=\regret^*\frac{\capex^*-k}{K-k}
=\regret^*-\frac d{K-k}.
\]
In part \textup{(iii)}, retain the proposition's $t,a,d$.
In both cases the candidate is nonnegative, continuously
differentiable, strictly decreasing and strictly concave on
$[0,\capex^*]$. Writing $s(k):=-w^{*\prime}(k)$, it satisfies
\[
0<s(k)<B/2<1,\qquad
B-1\leq s(t)=Bt/K\leq\alpha,\qquad
d=Bt(K-t)^2/K,\qquad w^*(\capex^*)=0.
\]

Fix any environment, any agent best response $k$, and any
first-best investment $k_P$. If $k_P>k$ and $k<\capex^*$,
right-sided optimality gives $\pi'_+(k)\leq s(k)$, since
$b'_+(k)\geq0$. Concavity therefore bounds regret by
\[
U(k):=w^*(k)+s(k)(K-k).
\]
If $k_P<k$, left-sided optimality gives
$-\pi'_-(k)\leq B-s(k)$; nondecreasing production also gives
$-\pi'_-(k)\leq1$. Concavity bounds regret by
\[
O(k):=w^*(k)+\min\{1,B-s(k)\}k.
\]
These arguments include $k=0$ in the first case and
$k=\capex^*$ in the second. The remaining case
$k=\capex^*<k_P$ gives regret at most
$\alpha(K-\capex^*)=\regret^*$ directly. If operating profit is
already first best, regret is $w^*(k)\leq\regret^*$.

On the first branch, $U(k)=\regret^*$. On the second branch,
which exists only in part \textup{(iii)}, $t\geq K/3$ and
\[
U'(k)=\frac{2aK}{k^3}-B
\leq B(K/t-3)\leq0.
\]
Thus $U(k)\leq\regret^*$ throughout. The parameter identities give
$O(t)=\regret^*$. For $k>t$, the second branch exists and
$s(k)\geq s(t)\geq B-1$, so its formula yields
$O(k)=\regret^*$.
For $k<t\leq K/3$, the function $w^*(k)+(B-s(k))k$ is
nondecreasing because
\[
\frac{d}{dk}\bigl[w^*(k)+(B-s(k))k\bigr]
=B-\frac{2dK}{(K-k)^3}
\geq\frac{B(K-3t)}{K-t}\geq0.
\]
It bounds $O(k)$ from above and equals $\regret^*$ at $t$.
If $t>K/3$, then $s(t)=B-1$, so for $k<t$ we instead have
$O(k)=w^*(k)+k$. This expression is increasing, since
$1-s(k)>0$, and again equals $\regret^*$ at $t$.
Hence every environment and every best response generate regret
at most $\regret^*$.

Finally, the two admissible environments
$((1+r)(1+s(t))k,0)$ and $((1+r)(1+s(t)-B)k,Bk)$ both induce the unique
choice $t$, including when $t=\capex^*$. Their regrets are
$w^*(t)+s(t)(K-t)=w^*(t)+(B-s(t))t=\regret^*$.
Together with the lower bounds, this proves the proposition.
\hfill\qedsymbol

\section{Proofs of the Comparative Statics}
\label{sec:statics-proof}

Write $\alpha=\hat A-1$, $\beta=\min\{B,1\}$, and
$z=\capex^*/K$. Derivatives below hold the other primitives, including $r$, fixed
and are taken within the stated parameter regions. Payment
comparisons hold investment $k$ fixed below the cap unless
otherwise specified.

\paragraph{Proof of Corollary \ref{cor:cap-statics}.}
For cases~(i) and~(ii), direct differentiation gives
\begin{equation}\label{eq:cap-statics-simple}
\begin{array}{c|ccc}
&\partial_A\capex^*&\partial_B\capex^*&\partial_K\capex^*\\ \hline
\text{(i)}&\dfrac{K}{(1+r)\hat A^2}&0&\alpha/\hat A\\[2pt]
\text{(ii)}&\dfrac{K}{(1+r)B}&-K\alpha/B^2&\alpha/B.
\end{array}
\end{equation}
For case~(iii), define $\tau=t/K$, $h=a/K^2$, and
$\delta=d/K^2$. These coefficients depend only on $B$:
\begin{equation}\label{eq:normalized-coefficients}
\begin{array}{c|cc}
&0<B\leq3/2&3/2<B<2\\ \hline
\tau&1/3&(B-1)/B\\[2pt]
h&B/54&(B-1)^2(2-B)/(2B^2)\\[2pt]
\delta&4B/27&(B-1)/B^2.
\end{array}
\end{equation}
The cap equation is $\alpha(1-z)=Bz/2+h/z$. Let
$s_z=B/2-h/z^2$ and $D=\alpha+s_z$. Since $z\geq\tau$,
\[
s_z\geq B/2-h/\tau^2=\max\{B/3,B-1\}>0.
\]
Implicit differentiation therefore yields
\begin{equation*}\label{eq:cap-statics-implicit}
\partial_A\capex^*=\frac{K(1-z)}{(1+r)D}>0,
\qquad
\partial_B\capex^*=-K\frac{z/2+h_B/z}{D}<0,
\qquad
\partial_K\capex^*=z>0,
\end{equation*}
where
\begin{equation*}\label{eq:h-derivative}
h_B=\begin{cases}
1/54,&B\leq3/2,\\
-\dfrac{(B-1)(B^2+B-4)}{2B^3},&3/2<B<2.
\end{cases}
\end{equation*}
The sign of the $B$ derivative is immediate when $B\leq3/2$.
For $B>3/2$,
\[
z^2/2+h_B\geq\tau^2/2+h_B
=\frac{(B-1)(2-B)}{B^3}>0.
\]

Where cases~(ii) and~(iii) meet, $\alpha=B/3$ and both give $z=1/3$;
their derivatives agree. At $B=3/2$, $h$ and $h_B$ agree
across the two formulas in \eqref{eq:normalized-coefficients}.
At the boundary $\hat A=B<2$, the compensated cap tends to
$z=(B-1)/B=\alpha/\hat A$. At $B=2$ and $\hat A\geq2$, $h$ tends
to zero and the cap equation again gives $z=\alpha/\hat A$.
Thus the cap is continuous everywhere. The derivative signs
extend to global strict monotonicity in $A$ and $K$ and weak
monotonicity in $B$, with strict decrease in $B$ whenever
compensation is useful. Proportionality to $K$ follows from
the independence of $z$ from $K$. Because $A$ and $r$ enter
only through $\hat A=A/(1+r)$,
\[
\partial_r\capex^*=-\frac{A}{1+r}\partial_A\capex^*<0.
\]
\hfill\qedsymbol

\paragraph{Proof of Corollary \ref{cor:delegation}.}
In case~(i), $B>1$ and $\capex^*=\capex_0$.
In case~(ii), $\alpha+\beta>B$: this is immediate for
$B\leq1$, and follows from $\hat A>B$ for $B>1$. Hence
$z=\alpha/B>\alpha/(\alpha+\beta)=\capex_0/K$.
In case~(iii),
\[
\frac{\alpha(1-z)}z
=\frac B2+\frac h{z^2}
\leq\frac B2+\frac h{\tau^2}
=\begin{cases}2B/3,&B\leq3/2,\\1,&B>3/2.\end{cases}
\]
For $B<3/2$, the final bound is strictly below $\beta$.
For $B\geq3/2$, $\hat A>B$ implies $z>\tau$, so the inequality
is strict and again gives a value below $\beta=1$.
Rearranging establishes $\capex^*>\capex_0$. Finally,
$\regret^*=\alpha(K-\capex^*)$ and
$\regret_0^*=\alpha(K-\capex_0)$ give the regret comparison.
\hfill\qedsymbol

\paragraph{Proof of Corollary \ref{cor:compensation-statics}.}
Set $x=k/K$, $v(x)=w^*(Kx)/K$, and
$\rho=\regret^*/K=\alpha(1-z)$. Throughout this proof,
$0\leq x<z$. In case~(ii),
\[
v(x)=\alpha(1-z)\frac{z-x}{1-x},\qquad z=\alpha/B<1/3.
\]
Differentiating gives
\begin{equation}\label{eq:wage-low-statics}
v_A(x)=\frac{z(2-3z)-(1-2z)x}{(1+r)(1-x)}>0,
\qquad
v_B(x)=-\frac{\alpha z}{B}\frac{1-2z+x}{1-x}<0.
\end{equation}
For the first sign, the numerator decreases in $x$ and
equals $z(1-z)>0$ at $x=z$. The second uses $z<1/3$.
Moreover, $v_B$ is strictly decreasing in $x$. Writing
$-w_k^*(k)=d/(K-k)^2$, where
$d=K^2\alpha(1-z)^2$, gives
\[
d_A=\frac{K^2(1-z)(1-3z)}{1+r}>0,
\qquad
d_B=\frac{2K^2\alpha z(1-z)}B>0.
\]
Thus both $A$ and $B$ steepen the schedule in this case.

In case~(iii),
\[
v(x)=\begin{cases}
\rho-\delta/(1-x),&0\leq x\leq\tau,\\
\rho-Bx/2-h/x,&\tau\leq x\leq z.
\end{cases}
\]
The coefficients $\tau,h,\delta$ are independent of $A$,
so the increase in compensation is the same at every common
investment:
\begin{equation*}\label{eq:wage-A-high}
v_A(x)=\rho_A=\frac{(1-z)s_z}{(1+r)D}>0.
\end{equation*}
For the private-benefit bound, define
\begin{equation*}\label{eq:wage-B-high}
\eta:=\rho_B=\frac{\alpha}{D}\left(\frac z2+\frac{h_B}z\right)>0,
\qquad v_B(x)=\eta-J(x),
\end{equation*}
where
\begin{equation*}\label{eq:J}
J(x)=\begin{cases}
\delta_B/(1-x),&0\leq x\leq\tau,\\
x/2+h_B/x,&\tau\leq x\leq z,
\end{cases}
\qquad
\delta_B=\begin{cases}
4/27,&B\leq3/2,\\
(2-B)/B^3,&3/2<B<2.
\end{cases}
\end{equation*}
The two expressions for $J$ agree at $\tau$. The first is
strictly increasing. On the second branch,
$J'(x)=1/2-h_B/x^2>0$: if $h_B\leq0$, the sign is
immediate; otherwise its minimum is at $\tau$, where it
equals $1/3$ for $B\leq3/2$ and
$(B^2-2)/(B(B-1))>0$ for $B>3/2$.
Consequently, $v_B$ is strictly decreasing in investment and
the schedule becomes steeper as $B$ rises. Also,
\[
J(0)=\delta_B,\qquad
\eta=\frac{\alpha}{\alpha+s_z}J(z)<J(z).
\]
If $\eta<\delta_B$, compensation falls at every common
investment; at equality, its derivative is zero only at
$k=0$. If $\eta>\delta_B$, the payment derivative changes
sign exactly once, from positive to negative, at the unique
$x\in(0,z)$ satisfying $J(x)=\eta$.

These comparisons also imply the corresponding finite-change
patterns when the comparison stays in a compensation region.
In particular, the difference between the schedules associated
with two values of $B$ is strictly decreasing over their common
menu: integrate the strictly decreasing derivative $w_B^*(k)$
over the intervening values of $B$. At the new, lower cap, the
new payment is zero and the old payment is strictly positive.
Thus, if the payment at zero weakly falls, payments strictly
fall at every positive investment in the common menu. If the
payment at zero rises, continuity gives exactly one crossing
strictly inside the common menu, with higher payments below
and lower payments above. When both contracts belong to
case~(ii), \eqref{eq:wage-low-statics} implies that every
payment in the common menu falls, including the payment at zero.

Finally, the representation $w^*(k)=Kv(k/K)$ gives exact
proportionality at a fixed expenditure share and
\begin{equation*}\label{eq:wage-K}
\left.\frac{\partial w^*(k)}{\partial K}\right|_k
=v(x)-xv_x(x)>0,
\end{equation*}
because $v(x)>0$ and $v_x(x)<0$ below the cap.
Strict concavity gives
\[
\left.\frac{\partial(-w_k^*(k))}{\partial K}\right|_k
=\frac{x}{K}v_{xx}(x)<0\qquad (k>0).
\]
Hence a larger budget flattens the schedule at a fixed positive
investment. At a moving joining point these slope comparisons
use one-sided derivatives, both of which have the stated sign.
All payment derivatives vanish in the interior of case~(i).
Within either compensation region, dependence on $A$ and $r$
is only through $\hat A$. Hence
\[
\partial_r w^*(k)=-\frac{A}{1+r}\partial_A w^*(k)<0.
\]
The case~(iii) change is constant over the common menu, while
the case~(ii) slope comparison follows by reversing the sign
of the corresponding $A$ comparison above.
The cap and payment-level formulas agree at the internal
compensation boundaries, but the selected schedule need not
be continuous at the pure-delegation boundary. The payment
comparisons therefore have the regional scope stated in the
corollary.
\hfill\qedsymbol

\paragraph{The productivity threshold for a payment increase.}
The sign test $\eta>\delta_B$ in case~(iii) is equivalent to
$\hat A>\hat A_c(B)$ for a unique threshold $\hat A_c(B)$ at each fixed
$B\in(0,2)$. To verify uniqueness, put
\[
\begin{gathered}
c=h/B=\tau^2(1-2\tau)/2,\qquad G(z)=z/2+c/z,\\
E(z)=G(z)+(1-z)G'(z),\qquad N(z)=z/2+h_B/z.
\end{gathered}
\]
Eliminating $\alpha$ using the cap equation gives
$\eta=G(z)N(z)/E(z)$. Both $N$ and $N'$ are positive,
and the derivative of $G/E$ has the sign of
\[
H(z)=\frac14-\frac{2c(1-z)}{z^2}-\frac{c^2}{z^4}.
\]
Since $H(\tau)=(1-\tau)(3\tau-1)\geq0$ and
$H'(z)=2c(2c+2z^2-z^3)/z^5>0$, $\eta$ increases
strictly with $z$, hence with $A$. At the lower boundary
of case~(iii), $\eta$ tends to $1/9<4/27=\delta_B$
when $B\leq3/2$, and to
$(2-B)/(2B^2)<(2-B)/B^3=\delta_B$ when $B>3/2$.
As $A\to\infty$, it tends to $1/2+h_B>\delta_B$.
The threshold therefore exists and is unique.

For $B\leq3/2$, it has the explicit form
$\hat A_c(B)=1+\lambda_0B$, where
\[
\lambda_0=\frac{z_0/2+1/(54z_0)}{1-z_0}
\simeq0.50045643,
\qquad
729z_0^4-162z_0^2-16z_0+9=0,\quad z_0\in(1/3,1).
\]
For $3/2<B<2$, it is equivalently characterized by
\[
(Bz_c^2+2h)(z_c^2+2h_B)
=2\delta_B(Bz_c^2+4hz_c-2h),
\qquad
\hat A_c(B)=1+\frac{Bz_c/2+h/z_c}{1-z_c},
\]
where the first equation has a unique solution
$z_c\in(\tau,1)$.

\section{Compensation Contingent on Realized Output}\label{app:outcome-compensation}

This appendix studies the output-contingent contracts discussed in Section \ref{sec:extensions}. We define the contract class, establish a lower bound on regret, and construct contracts that approach this bound in a specified parameter region.

Throughout, let $y=q(k)$ denote undiscounted date-1 output.

\paragraph{Contract space.}
An outcome-contingent contract consists of a nonempty closed menu
$\menu\subseteq[0,K]$ that contains zero, and a finite, nonnegative, jointly
upper-semicontinuous payment schedule $w(k,y)$ on
\[
\{(k,y):k\in\menu,\ 0\leq y\leq Ak\}.
\]
In environment $(q,b)\in\envir$, the manager chooses
\[
\operatorname{BR}(\menu,w;q,b)
:=
\operatorname*{arg\,max}_{k\in\menu}
\{\pi(k;q)+b(k)+w(k,q(k))\}.
\]
These conditions ensure that a best response exists in every
environment. Headquarters receives $\pi(k;q)-w(k,q(k))$. 
 As in Section \ref{sec:regret}, regret is the shortfall from the informed benchmark $V(q)$. Slightly abusing the notation, we define worst-case regret as
\[
\wcr(\menu,w):=
\sup_{(q,b)\in\envir}
\max_{k\in\operatorname{BR}(\menu,w;q,b)}
\bigl\{V(q)-\pi(k;q)+w(k,q(k))\bigr\}.
\]
The inner maximum imposes adverse tie breaking. We use a
supremum over environments because the regularity conditions do not guarantee that a worst-case environment
exists when output is contractible. 
All optimization statements below refer to the output-contingent contract class defined here.

When $B\geq2$, an uncompensated spending cap remains optimal. The argument in Section \ref{sec:joint} applies with $w(k)$ replaced by $w(k,0)$ in the zero-output environment: inducing less investment costs at least as much as the operating loss it avoids. Thus, writing $\capex=\max\menu$, worst-case regret is at least $\max\{\capex,\alpha(K-\capex)\}$. Minimizing this bound gives $\alpha K/\hat A$, attained by the uncompensated cap in Proposition \ref{prop_contract_k}(i). For $0<B<2$, we first establish a lower bound for every admissible output-contingent contract, then construct a family that approaches it when $0<\alpha\leq B/(2-B)$.

\begin{lemma}
\label{lem:out-lower-bound}
Suppose $0<B<2$ and $\alpha=\hat A-1=A/(1+r)-1>0$.
Every admissible outcome-contingent contract $(\menu,w)$ satisfies
\[
\wcr(\menu,w)\geq\frac{\alpha B}{2\alpha+B}K.
\]
\end{lemma}

\paragraph{Construction.}
Recall $\hat A=A/(1+r)$ and $\alpha=\hat A-1$.
Fix the investment cap and choose a positive payment threshold
satisfying
\begin{equation}
\label{eq:out-cap}
\capex=\frac{\alpha K}{\alpha+B/2},
\qquad
0<k_1<\frac{B}{B+2}\capex.
\end{equation}

\begin{proposition}
\label{prop:out-threshold}
Suppose $0<B<2$, $0<\alpha\leq\frac{B}{2-B}$, and \eqref{eq:out-cap} holds. Then
the contract $\menu=[0,\capex]$ with compensation
\begin{equation}
\label{eq:out-payment}
w(k,y)=
\begin{cases}
\displaystyle\frac B2(\capex-k),
&k_1\leq k\leq\capex,\quad y\leq(1+r)k,\\[4pt]
0,&\text{otherwise}
\end{cases}
\end{equation}
is admissible and has exact worst-case regret
\begin{equation}
\label{eq:out-policy-value}
\wcr([0,\capex],w)
=
\frac{\alpha B}{2\alpha+B}K
+\left(1-\frac B2\right)k_1.
\end{equation}
Consequently, throughout this parameter region,
\begin{equation}
\label{eq:out-value-bounds}
\inf_{\menu,w}\wcr(\menu,w)
=\frac{\alpha B}{2\alpha+B}K.
\end{equation}
Contracts of the form \eqref{eq:out-payment} approach this
infimum as $k_1$ decreases to zero.
\end{proposition}

\paragraph{Intuition.}
Compensation is available only when realized operating profit is nonpositive and investment reaches the positive payment threshold. At a positive investment, concavity implies that nonpositive operating profit rules out profitable expansion beyond that investment. Above the payment threshold and below the output cutoff, compensation falls by $B/2$ per additional unit of investment, offsetting half of the manager's maximal marginal private benefit.

At investments with positive operating profit, the contract withholds compensation. In particular, the productive environment $q(k)=Ak$, $b(k)=0$ induces investment at the cap and generates regret
\[
\alpha(K-\capex).
\]
Above the positive payment threshold, when realized operating profit is nonpositive, the proof bounds regret by
\[
\underbrace{\frac B2k}_{\text{overinvestment loss bound}}
+\underbrace{\frac B2(\capex-k)}_{\text{compensation}}
=\frac B2\capex.
\]
The parameter restriction also controls overinvestment at choices with positive operating profit. The cap balances the two bounds:
\[
\alpha(K-\capex)=\frac B2\capex.
\]
The positive threshold $k_1$ prevents payment at zero investment and introduces the additional worst-case cost $(1-B/2)k_1$. Shrinking this threshold makes the additional cost vanish while preserving admissibility at every step. The result establishes an infimum; it does not establish whether an admissible contract attains it.

\begin{corollary}[The value of output-contingent compensation]
\label{cor:output-value}
Suppose $\hat A>1$ and $B>0$. Let $\regret^*$ denote the
minimum worst-case regret with investment-only compensation,
as characterized in Proposition \ref{prop_contract_k}.
For the outcome-contingent contract class defined above,
\[
\frac89\regret^*
\leq
\inf_{\menu,w}\wcr(\menu,w)
\leq
\regret^*.
\]
Thus, allowing compensation to depend on realized output
can reduce worst-case regret by at most one ninth.
When $B\geq2$, the infimum equals $\regret^*$.
The factor $8/9$ is sharp.
\end{corollary}

\paragraph{Beyond the parameter restriction.}
When $0<B<2$ and $\alpha>\frac{B}{2-B}$, positive operating profit can coexist
with substantial overinvestment when a project has high initial
returns followed by a weak further marginal return. Allowing compensation
to vary with output above break-even ($y = (1+r)k$) may help address this
problem. One approach is to solve the partial differential
equations associated with binding regret bounds on their
respective regions, determining the boundaries between these
regions jointly with the payment schedule.
This is the exact analogy of Proposition \ref{prop_contract_k}. Nevertheless, the difficulty is
global: a solution to one branch equation may violate the other
regret bound, and local incentive conditions alone do not
establish managerial optimality. In particular, technologies
with high initial returns and low subsequent returns must be
considered alongside linear environments, and deviations across
payment boundaries must be checked. Upper semicontinuity permits
jumps, so smooth pasting is not required. Whether this approach can attain
or approximate $\frac{\alpha B}{2\alpha + B}K$ throughout the remaining
parameter region remains unresolved.

\begin{proof}[Proof of Lemma \ref{lem:out-lower-bound}]
Fix $\capex:=\max\menu\in[0,K]$. We first show that
\[
\wcr(\menu,w)\geq
\max\left\{\alpha(K-\capex),\frac B2\capex\right\}
\]
by constructing an admissible environment for each term.

First, consider
\[
q(k)=(1+r)\left(1-\frac B2\right)k,
\qquad b(k)=Bk.
\]
Since $0<B<2$ and $A>1+r$, both functions satisfy the
admissibility restrictions. Operating profit is
$\pi(k;q)=-Bk/2$, so the informed benchmark is $V(q)=0$.
For any managerial best response
$k_A\in\operatorname{BR}(\menu,w;q,b)$, headquarters' regret is
\[
\begin{aligned}
\regret(k_A;q,w)
&=0-\left[\frac{q(k_A)}{1+r}-k_A-w(k_A,q(k_A))\right]\\
&=\frac B2k_A+w(k_A,q(k_A)).
\end{aligned}
\]
The manager's payoff is
\[
\begin{aligned}
U_A(k_A)
&=\frac{q(k_A)}{1+r}-k_A+Bk_A+w(k_A,q(k_A))\\
&=\frac B2k_A+w(k_A,q(k_A))
=\regret(k_A;q,w).
\end{aligned}
\]
Because $\capex\in\menu$, optimality and nonnegative compensation
give, including when $k_A=\capex$,
\[
\regret(k_A;q,w)=U_A(k_A)
\geq\frac B2\capex+w(\capex,q(\capex))
\geq\frac B2\capex.
\]
Thus limited liability imposes a lower bound on regret:
compensation can restrain overinvestment only by replacing
the manager's forgone payoff from investing. For $0<B<2$ and
$\capex>0$,
this bound is strictly below the delegation-only overinvestment
bound $\min\{B,1\}\capex$.

Second, consider $q(k)=Ak$ and $b(k)=0$. Operating profit is
$\pi(k;q)=\alpha k$, and the informed headquarters chooses $K$.
Compensation may induce the manager to invest below the cap,
but every best response satisfies $k_A\leq\capex$. Therefore
\[
\regret(k_A;q,w)
=\alpha(K-k_A)+w(k_A,Ak_A)
\geq\alpha(K-\capex).
\]

Combining the two environments gives
\[
\wcr(\menu,w)\geq
\max\left\{\alpha(K-\capex),\frac B2\capex\right\}.
\]
The first term decreases with the cap and the second increases.
Their maximum is minimized where they coincide, at
\[
\capex=\frac{\alpha K}{\alpha+B/2}.
\]
Consequently,
\[
\min_{\capex\in[0,K]}
\max\left\{\alpha(K-\capex),\frac B2\capex\right\}
=\frac{\alpha B}{2\alpha+B}K,
\]
which proves the claimed lower bound.

\end{proof}
\begin{proof}[Proof of Proposition \ref{prop:out-threshold}]
    
\noindent\textbf{Step 1: Admissibility.}
The paid region in \eqref{eq:out-payment} is closed, and
compensation takes the higher adjacent value at both the
investment threshold and the output cutoff. Payments vanish
throughout a neighborhood of the origin and converge to zero
at the cap. Hence the schedule is finite, nonnegative, and
jointly upper semicontinuous. The contract is admissible,
and a managerial best response exists in every environment.

The cap and parameter restrictions imply
\begin{equation}
\label{eq:out-identities}
\alpha(K-\capex)
=\frac B2\capex
=\frac{\alpha B}{2\alpha+B}K,
\qquad
\frac{\alpha}{\hat A}\leq\frac B2.
\end{equation}
The threshold restriction also gives
\begin{equation}
\label{eq:out-threshold-bounds}
\left(1+\frac B2\right)k_1<\frac B2\capex,
\qquad
Bk_1<\frac B2\capex.
\end{equation}

Fix any environment $(q,b)\in\envir$, any
$k\in\operatorname{BR}([0,\capex],w;q,b)$, and any
$t\in\operatorname*{arg\,max}_{[0,K]}\pi(\cdot;q)$.
Write $\pi(v)=q(v)/(1+r)-v$.

\medskip
\noindent\textbf{Step 2: Underinvestment, $t\geq k$.}
If the selected payment is positive, then $k>0$ and
$q(k)/(1+r)\leq k$. Concavity implies that operating profit
is nonincreasing to the right of $k$. Optimality of $t$
therefore gives $\pi(t)=\pi(k)$, so regret equals compensation
and is at most $B\capex/2$.

If the selected payment is zero, nonnegative compensation
at every other investment implies that $k$ also maximizes
$\pi+b$ on $[0,\capex]$. Comparing with
$\min\{t,\capex\}$ and using monotonicity of $b$ gives
\[
\pi(k)\geq\pi\bigl(\min\{t,\capex\}\bigr).
\]
Regret is therefore zero if $t\leq\capex$. Otherwise,
\[
\regret(k;q,w)
\leq\pi(t)-\pi(\capex)
\leq\alpha(K-\capex)
=\frac B2\capex.
\]

\medskip
\noindent\textbf{Step 3: Overinvestment, $t<k$.}
If $k<k_1$, the selected payment is zero. Comparing with $t$
and using nonnegative compensation at $t$ yields
\[
\regret(k;q,w)
\leq b(k)-b(t)
\leq B(k-t)
\leq Bk_1
<\frac B2\capex.
\]

If $k=k_1$, monotonicity of output gives
\[
\pi(t)-\pi(k_1)\leq k_1-t\leq k_1.
\]
Consequently,
\[
\begin{aligned}
\regret(k_1;q,w)
&\leq k_1+\frac B2(\capex-k_1)\\
&=\frac B2\capex+\left(1-\frac B2\right)k_1.
\end{aligned}
\]

Now suppose $k_1<k\leq\capex$.
If $q(k)<(1+r)k$, the paid formula applies immediately to
the left of $k$, including when $k=\capex$ and the selected
payment is zero. The manager's left optimality condition gives
\[
\pi'_-(k)+b'_-(k)-\frac B2\geq0,
\qquad
-\pi'_-(k)\leq\frac B2.
\]
Concavity therefore implies
\[
\pi(t)-\pi(k)
\leq-\pi'_-(k)(k-t)
\leq\frac B2k.
\]
Adding compensation yields
\[
\regret(k;q,w)
\leq\frac B2k+\frac B2(\capex-k)
=\frac B2\capex.
\]

If instead $q(k)\geq(1+r)k$, then for every $v\leq k$,
\[
q(v)\leq\min\{Av,q(k)\}.
\]
Maximizing this envelope after discounting output and
subtracting investment gives
\[
\pi(t)-\pi(k)\leq k-\frac{q(k)}A.
\]
Using \eqref{eq:out-identities} and the bound
$w(k,q(k))\leq B(\capex-k)/2$, we obtain
\[
\begin{aligned}
\regret(k;q,w)
&\leq \frac B2(\capex-k)+k-\frac{q(k)}A\\
&\leq \frac B2(\capex-k)+\frac{\alpha}{\hat A}k\\
&\leq \frac B2\capex.
\end{aligned}
\]
This includes the output cutoff. Together, these cases
establish the upper bound in \eqref{eq:out-policy-value}
for every environment and every managerial best response.

\medskip
\noindent\textbf{Step 4: Exactness and an approximating sequence.}
Consider $q=b=0$. For $k<k_1$, managerial utility is $-k\leq0$.
For $k\geq k_1$, it equals
\[
\frac B2\capex-\left(1+\frac B2\right)k.
\]
This expression is strictly decreasing and is positive at
$k_1$ by \eqref{eq:out-threshold-bounds}. Thus $k_1$ is the
unique best response. The informed benchmark is zero, and
regret equals
\[
k_1+\frac B2(\capex-k_1)
=\frac B2\capex+\left(1-\frac B2\right)k_1.
\]
This proves \eqref{eq:out-policy-value}.

For $n=1,2,\ldots$, retain the same menu $[0,\capex]$ and set
\begin{equation}
\label{eq:out-sequence}
k_{1,n}:=\frac{B\capex}{(B+2)(n+1)},
\qquad
w_n(k,y)=
\begin{cases}
\displaystyle\frac B2(\capex-k),
&k_{1,n}\leq k\leq\capex,\quad y\leq(1+r)k,\\[4pt]
0,&\text{otherwise}.
\end{cases}
\end{equation}
Every threshold satisfies \eqref{eq:out-cap}, so each contract
is admissible. Its exact worst-case regret is
\begin{equation}
\label{eq:out-sequence-value}
\begin{aligned}
\wcr([0,\capex],w_n)
&=\frac{\alpha B}{2\alpha+B}K
  \left(1+\frac{2-B}{(B+2)(n+1)}\right)\\
&\downarrow\frac{\alpha B}{2\alpha+B}K.
\end{aligned}
\end{equation}
Lemma~\ref{lem:out-lower-bound} supplies the matching lower
bound, proving \eqref{eq:out-value-bounds}.

\end{proof}
\begin{proof}[Proof of Corollary \ref{cor:output-value}]
Every admissible investment-only contract remains admissible when output-contingent compensation is allowed. This establishes the upper bound.

First suppose $B\geq2$. Fix any admissible outcome-contingent
contract and write $\capex=\max\menu$. In the admissible
environment $q(k)=0$ and $b(k)=2k$, the manager's payoff
and headquarters' regret both equal $k+w(k,0)$. Since
$\capex$ is available and compensation is nonnegative,
every managerial best response generates regret of at
least $\capex$. The productive environment $q(k)=Ak$,
$b(k)=0$ also gives the lower bound
$\alpha(K-\capex)$. Consequently,
\[
\wcr(\menu,w)
\geq
\max\{\capex,\alpha(K-\capex)\}
\geq
\frac{\alpha K}{1+\alpha}
=\regret^*,
\]
where the last equality follows from
Proposition \ref{prop_contract_k}(i).
Together with the upper bound, this proves equality
when $B\geq2$.

Now suppose $0<B<2$. Lemma \ref{lem:out-lower-bound} gives
\[
\inf_{\menu,w}\wcr(\menu,w)
\geq
\frac{\alpha B}{2\alpha+B}K.
\]
It therefore suffices to show that
\[
\regret^*
\leq
\frac98\frac{\alpha B}{2\alpha+B}K.
\]
We verify this inequality in each case of
Proposition \ref{prop_contract_k}.

In case~(i), $B\geq\hat A=1+\alpha$, so
$\alpha\leq B-1$. Since $B<2$,
\[
\begin{aligned}
\frac{\regret^*}{\alpha BK/(2\alpha+B)}
&=\frac{2\alpha+B}{B(1+\alpha)}\\
&\leq\frac{3B-2}{B^2}\\
&=\frac98-\frac{(3B-4)^2}{8B^2}
\leq\frac98.
\end{aligned}
\]

In case~(ii),
\[
\begin{aligned}
\frac{\regret^*}{\alpha BK/(2\alpha+B)}
&=\left(1-\frac{\alpha}{B}\right)
  \left(1+\frac{2\alpha}{B}\right)\\
&=\frac98
  -2\left(\frac{\alpha}{B}-\frac14\right)^2
\leq\frac98.
\end{aligned}
\]

In case~(iii), the cap equation
\eqref{eq:konly-crossing} implies
\[
\begin{aligned}
\frac{\regret^*}{\alpha BK/(2\alpha+B)}
&=\frac{(2\alpha+B)(K-\capex^*)}{BK}\\
&=1+\frac{2a}{BK\capex^*}.
\end{aligned}
\]
Here $a=(B/2)t^2(1-2t/K)\geq0$,
$\capex^*\geq t$, and $t/K\in[1/3,1/2)$.
Hence
\[
1+\frac{2a}{BK\capex^*}
\leq
1+\frac{t}{K}\left(1-\frac{2t}{K}\right)
\leq
\frac{10}{9}
<
\frac98.
\]
The three cases establish the claimed lower bound.

Finally, take $B=1$ and $\alpha=1/4$.
Proposition \ref{prop_contract_k}(ii) gives
$\regret^*=3K/16$. These parameters satisfy the conditions
of Proposition \ref{prop:out-threshold}, which gives
\[
\inf_{\menu,w}\wcr(\menu,w)
=\frac{K}{6}
=\frac89\regret^*.
\]
Thus the factor $8/9$ cannot be increased. This establishes
sharpness in terms of the infimum without asserting that
an outcome-contingent contract attains it.
\end{proof}
\end{document}